\documentclass[11pt,oneside]{amsart} % remove `oneside' in final version
\usepackage{amscd,amssymb,amsxtra}
\usepackage[mathscr]{eucal}% \usepackage{mathrsfs}  % \mathscr{A}
\usepackage{mathabx}
\usepackage{enumerate}
\usepackage{comment}
\usepackage{color}
\usepackage{hyperref}
\makeatletter
\let\@secnumfont\bfseries
\def\section{\@startsection{section}{1}%
  \z@{4\linespacing\@plus\linespacing}{\linespacing}%
  {\bfseries\centering}}
\def\introsection{\@startsection{section}{1}%
  \z@{3\linespacing\@plus\linespacing}{\linespacing}%
  {\bfseries\centering}}
 \def\subsection{\@startsection{subsection}{2}%
   \z@{1.25\linespacing\@plus.7\linespacing}{.5\linespacing}%
   {\normalfont\bfseries}}
\makeatother

\theoremstyle{definition}

\newtheorem*{example*}{Example}
\newtheorem*{problem*}{Problem}
\newtheorem*{exercise*}{Exercise}
\newtheorem*{question*}{Question}

\theoremstyle{remark}

\newtheorem{remark}[equation]{Remark}

\newtheorem*{note*}{Note}
\newtheorem*{notation*}{Notation}
\newtheorem*{remark*}{Remark}

\theoremstyle{plain}
\newtheorem{definition}[equation]{Definition}
\newtheorem{theorem}[equation]{Theorem}
\newtheorem{corollary}[equation]{Corollary}
\newtheorem{lemma}[equation]{Lemma}
\newtheorem{proposition}[equation]{Proposition}

\newtheorem{proposal}[equation]{Proposal}

\newtheorem*{definition*}{Definition}
\newtheorem*{theorem*}{Theorem}
\newtheorem*{corollary*}{Corollary}
\newtheorem*{lemma*}{Lemma}
\newtheorem*{proposition*}{Proposition}
\newtheorem*{conjecture*}{Conjecture}
\newtheorem*{claim*}{Claim}
\newtheorem*{proposal*}{Proposal}
\newtheorem*{conclusion*}{Conclusion}

\numberwithin{equation}{section}

\definecolor{refkey}{rgb}{0,.6,.4}

\renewcommand{\:}{\colon}
\DeclareMathOperator{\End}{End}

\DeclareMathOperator{\Hom}{Hom}
\DeclareMathOperator{\id}{id}

\DeclareMathOperator{\pt}{pt}

\newcommand{\RP}{{\mathbb R\mathbb P}}
\newcommand{\RR}{{\mathbb R}}
\newcommand{\TT}{\mathbb T}
\DeclareMathOperator{\Spin}{Spin}

\newcommand{\ZZ}{{\mathbb Z}}

\newcommand{\chiup}{\raise.5ex\hbox{$\chi$}}
\newcommand{\cir}{S^1}

\newcommand{\inv}{^{-1}}
\newcommand{\mstrut}{^{\vphantom{1*\prime y}}}

\DeclareMathOperator{\rank}{rank}

\newcommand{\temsquare}{\raise3.5pt\hbox{\boxed{ }}}

\newcommand{\zmod}[1]{\ZZ/#1\ZZ}

\newcommand{\zt}{\zmod2}

\usepackage[all,2cell,dvips]{xy}\renewcommand{\cir}{\ensuremath{S^1}}
\usepackage{color}

\DeclareMathOperator{\Cliff}{Cliff}
\DeclareMathOperator{\Det}{Det}
\DeclareMathOperator{\Euler}{Euler}
\DeclareMathOperator{\Image}{Image}
\DeclareMathOperator{\Pfaff}{Pfaff}
\DeclareMathOperator{\Pic}{Pic}
\DeclareMathOperator{\Pin}{Pin}
\DeclareMathOperator{\hol}{hol}
\DeclareMathOperator{\pfaff}{pfaff}
\newcommand{\Luniv}{L^{\textnormal{univ}}}
\newcommand{\Opinm}{\Omega^{\textnormal{Pin${}^-$}}_2}
\newcommand{\Picdz}[1]{\Pic^0_\nabla (#1)}
\newcommand{\Picd}[1]{\Pic_\nabla (#1)}
\newcommand{\Pinm}{\Pin^-}
\newcommand{\RZ}{\RR/\ZZ}
\newcommand{\Spinstr}{\mathscr{S}}
\newcommand{\Xa}{X_{a(\cb)}}
\newcommand{\Xw}{X_w}
\newcommand{\al}{\alpha \mstrut_\ell}
\newcommand{\am}{\alpha ^-}
\newcommand{\beq}{[\cb]}

\newcommand{\cH}{\widecheck{H}}
\newcommand{\cR}{\widecheck{R}}
\newcommand{\cb}{\check{\beta }}
\newcommand{\ce}{\check{\eta}}
\newcommand{\chh}{\widecheck{h}}
\newcommand{\cth}{\check{\theta }}
\newcommand{\gpd}{/\!/\,} 
\newcommand{\hM}{\widehat{M}}
\newcommand{\hS}{\widehat{\Sigma }}
\newcommand{\hp}{\hat\pi }
\newcommand{\hs}{\hat\sigma }
\newcommand{\pinm}{$\text{pin}^-$}
\newcommand{\pinp}{$\text{pin}^+$}
\newcommand{\sBSO}{\sB_{SO}}
\newcommand{\sB}{\mathcal{B}}
\newcommand{\sF}{\mathcal{F}}
\newcommand{\sH}{\mathcal{H}}

\newcommand{\sPm}{\mathscr{P}^-}
\newcommand{\sT}{\mathcal{T}}
\newcommand{\son}{\mathfrak{s}\mathfrak{o}_n}
\newcommand{\sstr}{\sB_{\textnormal{Spin}}}

\newcommand{\tOn}{\widetilde{O}_n}
\newcommand{\tOt}{\widetilde{O}_{10}}
\newcommand{\tR}[1]{\tau^R(#1)}

\newcommand{\tf}{\tilde\phi }
\newcommand{\tg}{\tilde{g}}
\newcommand{\thol}{\widetilde{\hol}}
\newcommand{\tko}{\widetilde{ko}\,}

\newtheorem{supposition}[equation]{Supposition}

\begin{document}

\abovedisplayskip18pt plus4.5pt minus9pt
\belowdisplayskip \abovedisplayskip
\abovedisplayshortskip0pt plus4.5pt
\belowdisplayshortskip10.5pt plus4.5pt minus6pt
\baselineskip=15 truept
\marginparwidth=55pt

\renewcommand{\labelenumi}{(\roman{enumi})}

%**end of header

% lasteq@ 63
% lastsec@  7
% lastthm@ 47
% lastfig@000

%$$\boxed{\boxed{\text{PRELIMINARY VERSION}}}$$\par\vskip 2pc % omit in final

 \title{Spin structures and superstrings} %% replace \today with
                                                %% short title in final version 
 \thanks{Report \#: UTTG-07-10}
 \author[J. Distler]{Jacques Distler} 
 \thanks{The work of J.D. is supported by the National Science Foundation
under grant PHY-0455649 and a grant from the US-Israel Binational Science
Foundation.}   
 \address{Theory Group, Department of Physics, and Texas Cosmology Center \\
University of Texas \\ 1 University Station C1600\\ Austin, TX 78712-0264}
 \email{distler@golem.ph.utexas.edu}

 \author[D. S. Freed]{Daniel S.~Freed}
 \thanks{The work of D.S.F. is supported by the National Science Foundation
under grant DMS-0603964}
 \address{Department of Mathematics \\ University of Texas \\ 1 University
Station C1200\\ Austin, TX 78712-0257}
 \email{dafr@math.utexas.edu}

 \author[G. W. Moore]{Gregory W.~Moore}
 \thanks{The work of G.W.M. is supported by the DOE under grant
DE-FG02-96ER40949.}
 \address{NHETC and Department of Physics and Astronomy \\
Rutgers University \\ Piscataway, NJ 08855--0849}
 \email{gmoore@physics.rutgers.edu}
 \thanks{We also thank the Aspen Center for Physics for providing a
stimulating environment for many discussions related to this paper.}
% \thanks{Report numbers: } 
 \dedicatory{To Isadore Singer on the occasion of his 85th birthday}
 \date{July 26, 2010}
 \begin{abstract} 
 In superstring theory spin structures are present on both the 2-dimensional
worldsheet and 10-dimensional spacetime.  We present a new proposal for the
$B$-field in superstring theory and demonstrate its interaction with
worldsheet spin structures.  Our formulation generalizes to orientifolds,
where various twistings appear.  A special case of the orientifold worldsheet
$B$-field amplitude is a $KO$-theoretic construction of the $\zmod8$-valued
Kervaire invariant on \pinm\ surfaces.
 \end{abstract}
\maketitle

%\pagestyle{myheadings}   % omit in final
%\markboth{PRELIMINARY VERSION (\today)}{PRELIMINARY VERSION (\today)}  % omit

The Type II superstring in the NSR formulation is a theory of maps from a
closed surface~$\Sigma $---the worldsheet---to a 10-manifold~$X$---spacetime.
The spin structures of the title are present on both the worldsheet and the
spacetime.  Their roles have been explored in many works; a sampling of
references includes~\cite{GSO1, GSO2, SS1, SS2, R, SW, DH, AgMV, AgGMV, AW}.
In this paper we identify several new phenomena which are intimately related
to a new Dirac quantization condition for the $B$-field
(Proposal~\ref{thm:1}).  For example, in our approach the $B$-field amplitude
depends on the worldsheet spin structure.  In particular, the distinction
between Types~IIB and~IIA is encoded in the $B$-field and the worldsheet
$B$-field amplitude includes the usual signs in the sum over spin structures.
In another direction we answer the question: How does the spacetime spin
structure impact the worldsheet theory in the lagrangian formulation?  It
turns up in the definition of the partition function of worldsheet fermions,
i.e., in computing the pfaffian of the Dirac operator on~$\Sigma $.  For
\emph{orientifolds} of the Type~II superstring, including the Type~I
superstring, there are several new features.  For example, we define
precisely the twisted notions of spin structure needed on~$\Sigma $ and
on~$X$.  We also consider the worldsheet $B$-field amplitude and the
partition function of worldsheet fermions.  It turns out that each is
anomalous and that these anomalies cancel.  That anomaly cancellation is the
subject of a future paper~\cite{DFM2}; here we are content to motivate that
work and consider some special cases.
 
Evidently, these spin structure considerations are closely tied to the
$B$-field~$\cb$, with which we begin in~\S\ref{sec:1}.  Quite generally,
Dirac quantization of charges and fluxes is implemented by generalized
cohomology theories.  For the oriented bosonic string the $B$-field has a
flux quantized by~$H^3(X;\ZZ)$.  We locate the \emph{superstring} $B$-field
quantization condition in a \emph{generalized} cohomology theory~$R$ which is
a truncation of connective $KO$-theory.  Then the $B$-field is modeled in the
\emph{differential} cohomology group~$\cR\inv (X)$ using the general
development of differential cohomology in~\cite{HS}.  In~\S\ref{sec:2} we
take up the integral of $\phi ^*\cb$ on the worldsheet~$\Sigma $ for maps
$\phi \:\Sigma \to X$.  The presence of $KO$-theory suggests the dependence
on worldsheet spin structures.  We show how the standard $\zt$-valued
quadratic function on spin structures~\cite{A1} is embedded in the $B$-field
amplitude, leading to the distinction between Types~IIB and~IIA.  A
generalization of the Scherk-Schwarz construction~\cite{SS1, SS2} is
also part of our $B$-field amplitude.  Orbifolds (in the sense of string
theory) and orientifolds are introduced in~\S\ref{sec:3}.  To accommodate the
former we allow~$X$ to be an orbifold (in the sense of differential
geometry); the orientifold is encoded in a double cover $\pi \:X_w\to X$ of
orbifolds.  The $B$-field~$\cb$ is now quantized by the $R$-cohomology of the
Borel construction applied to~$X$, with local coefficients determined by the
double cover~$\pi $ (Proposal~\ref{thm:13}).  The integral of~$\phi ^*\cb$ is
taken up in~\S\ref{sec:4}.  We posit a spin structure on the orientation
double cover $\hp\:\hS\to\Sigma $ of the worldsheet.  In case this refines
and is refined to a \pinm\ structure the integral of~$\phi ^*\cb$ may be
easily defined.  For a certain $B$-field this yields a $KO$-theoretic
construction of the $\zmod8$-valued Kervaire invariant on \pinm\
surfaces~\cite{Bro}, \cite[\S3]{KT}.  For a general (non-\pinm)
spin structure on~$\hS$ the $B$-field amplitude is anomalous~\eqref{eq:59};
its definition is postponed to~\cite{DFM2}.  In \S\ref{sec:5} we prove a
formula for the pfaffian line of the Dirac operator in a related
one-dimensional supersymmetric quantum mechanical model, the one which
computes the index of the Dirac operator.  That formula is a categorified
index theorem in low dimensions.  We see explicitly how the spin structure on
spacetime enters.  This result is included here as motivation
for~\cite{DFM2}, where we take up the analogous problem on the
two-dimensional worldsheet.  The precise nature of the spin structure on
spacetime for orientifolds is the subject of~\S\ref{sec:6}.  It is a twisted
version of the usual notion of spin structure, where the twisting depends on
the orientifold double cover $\pi \:X_w\to X$ as well as the $B$-field~$\cb$.
 
The telegraphic pr\'ecis~\cite{DFM1} outlines many aspects of orientifold
theory.  This is the first of several papers which expatiate on this
r\'esum\'e.  These papers provide motivation, give precise definitions,
develop some background mathematics, state and prove the main theorems, and
give applications to physics.  The geometry of the $B$-field is further
developed in subsequent papers.  In~\cite{DFM2} we build a geometric model of
~$\cR\inv (X)$.  The geometric model is used in~\cite{DFM3} to twist
$K$-theory and its cousins, thus defining the home of the Ramond-Ramond field
on~$X$.  The $B$-field is a twisting of $K$-theory.  This relation to
twistings of $K$-theory is one of the main motivations for the choice of
Dirac quantization condition for the $B$-field.
 
The ideas here touch on many mathematical works of Isadore Singer: among
others his recent paper~\cite{HS} on quadratic forms and generalized
differential cohomology, his many contributions to index theory and the
geometry of Dirac operators, and even his use of frame bundles to express
geometric structures on manifolds~\cite{S}.  Beyond that his prescient
recognition 30~years ago of the role that theoretical high energy physics
would play in late $20^{\textnormal{th}}$~century and early
$21^{\textnormal{st}}$~century mathematics has had enormous influence on the
entire field.  
 
We thank Andrew Blumberg, Mike Hopkins, Isadore Singer, and Edward Witten for
helpful discussions.

  \section{$B$-fields and generalized differential cohomology}\label{sec:1}
% lastsubsec@  3

In classical physics an abelian gauge field is determined by its field
strength~$F$, a closed differential form on spacetime~$X$.  The archetype is
the Maxwell electromagnetic field, a closed 2-form in 4~spacetime
dimensions.\footnote{The word `gauge' in `classical gauge theory' applies
when we identify~$\Omega ^2(X)_{\textnormal{exact}}\cong \Omega
^1(X)/\Omega^1(X)_{\textnormal{closed}}$.}  Abelian gauge theories include an
electric current~$j$, which in Maxwell theory is a closed 3-form with compact
support on spacelike hypersurfaces.  The de Rham cohomology class of~$F$ is
called the \emph{classical flux}\footnote{Our usage of `flux' is not entirely
standard.} and the de Rham cohomology class of~$j$ the \emph{classical
charge}.  (The latter is taken with compact supports in spatial directions.)
In quantum theories Dirac's quantization principle constrains these classical
fluxes and charges to full lattices inside the appropriate de Rham cohomology
groups.  For example, the quantum Maxwell electromagnetic flux is constrained
to the image of~$H^2(X;\ZZ)$ in~$H^2(X;\RR)\cong H^2_{dR}(X)$.  It is natural
to refine the flux to the abelian group~$H^2(X;\ZZ)$.  Indeed, in the quantum
theory the Maxwell electromagnetic field is modeled as a connection on a
principal circle bundle $P\to X$, and the flux is the topological equivalence
class of~$P$.  The electric charge is then refined to~$H^3(X;\ZZ)$ (with
appropriate supports), and there is a magnetic charge in the quantum theory
as well.  This leads to the notion that for any abelian gauge field, charges
and fluxes lie in abelian groups which are cohomology groups of spacetime.
It is a relatively recent discovery that \emph{generalized} cohomology groups
may occur.  Spacetime anomaly cancellation~\cite{GHM, MM} led to the
proposal, further elaborated in~\cite{W2}, that the Ramond-Ramond charges in
superstring theory are properly quantized by K-theory, at least in the large
distance and weak coupling limit.  Similarly, the fluxes are also quantized
by K-theory~\cite{FH, MW}.  In general, to quantize a classical abelian gauge
field one must choose a generalized cohomology group which reproduces the
appropriate de Rham cohomology vector space after tensoring over the reals.
The choice of cohomology theory is an input.  There are many physical
considerations which motivate the choice and can be used to justify it.
See~\cite[Part~3]{F1}, \cite{W3, OS, M} for leisurely expositions of these
ideas, including some examples.

In string theory, spacetime~$X$ is a smooth manifold whose dimension is~26
for the bosonic string and~10 for the superstring.\footnote{We use
`superstring' as a shorthand for `Type~II superstring' in a sigma model
formulation.}  In each case there is an abelian gauge field---the
``$B$-field''---whose field strength is a closed 3-form $H\in \Omega ^3(X)$.
Dirac's principle applies and we must locate the quantum flux in a cohomology
group.  The most natural choice applies a simple degree shift to the Maxwell
case.

        \begin{supposition}[]\label{thm:6}
 The flux of the oriented bosonic string $B$-field lies in~$H^3(X;\ZZ)$. 
        \end{supposition}

\noindent
 This supposition is certainly well-established~\cite{RW}.  In this
section we make a new proposal for the oriented superstring.

 \subsection{The cohomology theory~$R$}\label{subsec:1.1}

Let $ko$~denote connective $KO$-theory.  One construction~\cite{Se} starts
with the symmetric monoidal category of real vector spaces and applies a
de-looping machine to construct an infinite loop structure on its classifying
space.  More concretely, $ko$~is the real version of $K$-theory developed
in~\cite{A2} before inverting the Bott element; for any space~$M$ the abelian
groups~$ko^{q}(M)$ vanish for~$q>0$ and $ko^{-q}(M)\cong KO^{-q}(M)$
for~$q\ge0$.  Define the Postnikov truncation\footnote{We use the version of
Postnikov truncation for connective $E^{\infty}$-ring spectra~\cite{B}.  The
notation `$R$' for a multiplicative spectrum is generic, ergo uninformative,
but it would be cumbersome to use `$ko\langle 0\cdots4 \rangle$' instead.}
  \begin{equation}\label{eq:1}
     R:=ko\langle 0\cdots 4 \rangle .
  \end{equation}
Then $R$~is a generalized multiplicative cohomology theory, more precisely an
$E^{\infty}$-ring spectrum.  Its nonzero homotopy groups are
  \begin{equation}\label{eq:2}
     \{\pi _0,\pi _1,\pi _2,\pi _3,\pi _4\}(R)\cong \{\ZZ,\zt,\zt,0,\ZZ\},
  \end{equation}
a truncated Bott song.  These are also the nonzero $R$-cohomology groups of a
point and they occur in nonpositive degrees, as $R^{-q}(\pt)=\pi _{q}(R)$.
If we represent the theory as a (loop) spectrum~$\{R_{p}\}_{p\in \ZZ}$, so
that for any space~$M$ and~$q\ge0$ we compute $R^{-q}(M)=[M,R_{-q}]$ as the
abelian group of homotopy classes of maps into the space~$R_{-q}$, then
\eqref{eq:2}~are the homotopy groups of the space~$R_0$.

Here is our new proposal for the $B$-field in superstring theory.  Let $X$~be
a smooth 10-dimensional manifold which plays the role of spacetime in the
superstring.

        \begin{proposal}[]\label{thm:1}
 The flux of the oriented superstring $B$-field~$\cb $ lies in~$R\inv (X)$.
        \end{proposal}

\noindent
 As a first check we note that the nonzero homotopy groups of the
space~$R_{-1}$ are
  \begin{equation}\label{eq:3}
      \{\pi _0,\pi _1,\pi _2,\pi _3\}(R_{-1})\cong \{\zt,\zt,0,\ZZ\} ,
  \end{equation}
so after tensoring with the reals we obtain the Eilenberg-MacLane
space~$K(\RR,3)$ which computes real cohomology in degree~3.  This is as it
should be: the classical fluxes of the classical field~$H$ lie in the
degree~3 de Rham cohomology of the manifold~$X$.  We explore some physical
consequences of the nonzero torsion homotopy groups in~\S\ref{sec:2}.  

We record the exact sequence of abelian groups
  \begin{equation}\label{eq:4}
     0 \longrightarrow H^3(M;\ZZ)\longrightarrow R^{-1}(M)\xrightarrow{(t,a)}
     H^0(M;\zt)\times H^1(M;\zt)\longrightarrow 0 
  \end{equation}
which follows from the Postnikov tower (see~\eqref{eq:3}) and holds for any
space~$M$.  There is \emph{not} a corresponding exact sequence of cohomology
theories; the $k$-invariant between the bottom two homotopy groups is
nonzero.  The quotient group in~\eqref{eq:4} is more properly regarded as the
group of equivalence classes of $\zt$-graded real line bundles (equivalently:
$\zt$-graded double covers) over~$M$.  The exact sequence~\eqref{eq:4}
immediately implies
  \begin{equation}\label{eq:5}
     R\inv (\pt)\cong \zt,
  \end{equation}
and we can identify a generator with the nonzero element $\eta \in ko\inv
(\pt)\cong KO^{-1}(\pt)\cong \zt$.

There is a natural splitting of~\eqref{eq:4} as sets (not as abelian groups).
To construct it we interpret the quotient group as the group of $\zt$-graded
real line bundles and apply the following lemma.

        \begin{lemma}[]\label{thm:29}
 Let $V\to M$ be a real vector bundle over a space~$M$ and $[V]\in R^0(M)$
its equivalence class under the map $ko^0(M)\to R^0(M)$.  Then for $\eta
[V]\in R\inv (M)$ we have
  \begin{align}\label{eq:49}
     t\bigl(\eta [V] \bigr) &= \rank(V)\pmod2 \\ \label{eq:50} a\bigl(\eta
     [V] \bigr) &=w_1(V), 
  \end{align}
where $\rank(V)\:\pi _0M\to\ZZ$ is the rank.
        \end{lemma}

        \begin{proof}
 The map~$t$ in~\eqref{eq:4} is determined on the 0-skeleton~$M^0$ of~$M$,
and $V$~is equivalent to~$\rank(V)$ in~$ko^0(M^0)$.  This
reduces~\eqref{eq:49} to the assertion~$t(\eta )=1$, which is essentially the
isomorphism~\eqref{eq:5}.  The map~$a$ in~\eqref{eq:4} is determined on the
1-skeleton, and as $a(\eta )=0$ we can replace~$V$ by its reduced determinant
line bundle~$\Det V-1$, which is equivalent to~$V-\rank V$ in the reduced
group~$\widetilde{ko}^0(M^1)$.  Hence it suffices to prove~\eqref{eq:50} for
the universal real line bundle $\Luniv\to\RP^{\infty}$.  Identify~$ko\inv
(\pt)\cong \widetilde{ko}^0(\RP^1)$ and represent~$\eta $ by the reduced
M\"obius line bundle $(H-1)\to\RP^1$.  Then $\eta [\Luniv]$~is represented by
the external tensor product~$(H-1)\otimes \Luniv\to \RP^1\times \RP^\infty $.
To compute the $a$-component in~\eqref{eq:4} we restrict to the 1-skeleton
$\RP^1\subset \RP^{\infty}$, over which $\Luniv$~is identified with~$H$.
Again since $a(\eta )=0$ we may replace $(H-1)\otimes (H-1)\to\RP^1\times
\RP^1$, and this represents $\eta ^2\in ko^{-2}(\pt)$, which is nonzero.
This proves $\eta [H-1]$~is the nonzero class in $R\inv (\RP^1/\RP^0)\cong
H^1(\RP^1/\RP^0;\zt)$.  Therefore $a\bigl(\eta [H-1] \bigr)$, hence also
$a\bigl(\eta [\Luniv] \bigr)$, is nonzero.
        \end{proof}

 \subsection{Generalized differential cohomology and superstring
$B$-fields}\label{subsec:1.2} 

Semi-classical models of abelian gauge fields, which appear as background
fields or as inputs to a functional integral, combine the local information
of the classical field strength with the integrality of the quantum flux.  As
mentioned earlier the model for the Maxwell field is a circle bundle with
connection: its curvature is the classical field strength and its Chern class
the quantum flux.  Notice that there are nontrivial connections for which
both of these vanish.  In other words, the combination of classical field
strength and quantum flux do \emph{not} determine the semi-classical gauge
field.  Equivalence classes of Maxwell fields, thus of circle connections, on
any smooth manifold~$M$ form a topological abelian group~$\Picd M$, a
differential-geometric analog of the Picard group in algebraic geometry.
Its group of path components is
  \begin{equation}\label{eq:6}
     \pi _0\Picd M\cong H^2(M;\ZZ) 
  \end{equation}
the group of equivalence classes of circle bundles.  The map $\Picd M\to \pi
_0\Picd M$ forgets the connection.  The torus~$H^1(M;\RZ)$ of equivalence
classes of flat circle connections acts freely on the identity
component~$\Picdz M$ by tensor product, and the quotient
  \begin{equation}\label{eq:7}
     \Picdz M\to \Omega ^2_{\textnormal{exact}}(M) 
  \end{equation}
is the vector space of exact 2-forms.  Other components of~$\Picd M$ are also
total spaces of principal $H^1(M;\RZ)$-bundles; the bases are affine
translates of~$\Omega ^2_{\textnormal{exact}}(M)$ in the topological vector
space of closed 2-forms, affine spaces of closed forms with a fixed de Rham
cohomology class in the lattice $\Image\bigl(H^2(M;\ZZ)\to H^2(M;\RR)
\bigr)$.
 
Cheeger-Simons~\cite{CS} introduced topological abelian groups~$\cH^q(M)$ for
all integers~$q$ which generalize $\cH^2(M)\cong \Picd M$.  The
group~$\cH^1(M)$ is the group of smooth maps $M\to\TT$ into the circle group.
The group~$\cH^3(M)$ may be modeled as equivalence classes of $\TT$-gerbes
with connection or bundle gerbes~\cite{Br, Hi, Mu}.  The
original definition of~$\cH^q(M)$ is in terms of the integral over smooth
singular $(q-1)$-cycles, generalizing the holonomy of a $\TT$-connection
around a loop.  There is an alternative approach using sheaves, modeled after
a construction of Deligne in algebraic geometry.  Hopkins-Singer~\cite{HS}
provide two important supplements.  First, they define differential
cohomology groups~$\chh^{\bullet }(M)$ for any cohomology theory~$h$.
Second, they define spaces\footnote{In fact, they define simplicial sets.  We
use the moniker `points' for its 0-simplices.}~$\chh_{p }(M)$ such that $\pi
_0\chh_p(M)\cong \chh^p(M)$.  Thus points of~$\chh_{p}(M)$ may be considered
as geometric objects whose equivalence class lies in~$\chh^p(M)$, just as a
circle bundle with connection has an equivalence class in~$\Picd M$.  For the
specific cohomology theory~$R$ in~\eqref{eq:1} fix a singular cocycle $\iota
\in C^3(R_{-1};\RR)$ whose cohomology class is a normalized generator
of~$H^3(R_{-1};\RR)$.  Then a point of degree~$-1$ is a triple $(c,h,\omega
)$, where
  \begin{equation}\label{eq:8}
     \begin{aligned} c&\:M\longrightarrow R_{-1} \\ h &\in C^2(M;\RR) \\
     \omega &\in \Omega  ^3(M)\end{aligned} 
  \end{equation}
and $h$~satisfies $\delta h=\omega -c^*\iota $.  (It follows that $d\omega=0
$.)  We give $\cR^p(M)$ the structure of a topological abelian group for which
  \begin{equation}\label{eq:26}
     \pi _0\cR^p(M)\cong R^p(M) 
  \end{equation}
and each component is a principal $R^{p-1}(M;\RZ)$-bundle over an affine
space of closed differential forms.

The preceding discussion leads to corollaries of Supposition~\ref{thm:6} and
Proposal~\ref{thm:1}:
  \begin{align}\label{eq:60}
     &\textnormal{\it The oriented bosonic string $B$-field $\cb$ is a point in
     $\cH_{3}(X)$.} \\ 
     &\textnormal{\it The oriented superstring $B$-field $\cb$ is a point in
     $\cR_{-1}(X)$.} \label{eq:9}
  \end{align} 
In~\cite{DFM2} we give a concrete differential-geometric model of the
superstring $B$-field, whereas the model in terms of the spaces~$\cR_{p}(X)$
is more homotopy-theoretic.  In any case for the purposes of this paper we
only need the equivalence class~$\beq\in \cR^{-1}(X)$ of~$\cb $.  We remark
that $\cb$~determines $\beta\in R_{-1}$ whose equivalence class is~$[\beta
]\in R\inv (X)$; see~\eqref{eq:23} below.  Then using~\eqref{eq:4} we define
  \begin{equation}\label{eq:19}
     \bigl(t(\cb),a(\cb) \bigr)\in H^0(X;\zt)\times H^1(X;\zt). 
  \end{equation}
The physical significance of~\eqref{eq:19} is explained in subsequent
sections. 
 
We record the following exact sequences, which are specializations to the
case at hand of general facts about differential cohomology and hold for any
smooth manifold~$M$:
  \begin{align}     \label{eq:10}  
     0 \longrightarrow R^{-(q+1)}(M;\RZ)\longrightarrow
      \cR^{-q}(M)\longrightarrow \Omega ^{4-q}_{\ZZ}(M)\longrightarrow 0 \\
      \label{eq:23}
      0\longrightarrow \Omega ^{3-q}(M)/\Omega^{3-q}_\ZZ (M)\longrightarrow
      \cR^{-q}(M)\longrightarrow R^{-q}(M)\longrightarrow 0
  \end{align}
Here $q=1,2,3$ and $\Omega ^{4-q}_\ZZ(M)$~ denotes the space of closed forms
with integral periods.  In particular, it follows from these sequences
and~\eqref{eq:5} that
  \begin{equation}\label{eq:11}
     R^{-2}(\pt;\RZ)\cong \cR^{-1}(\pt) \cong R^{-1}(\pt)\cong \zt. 
  \end{equation}
The nonzero element~$\ce$ of~\eqref{eq:11} pulls back to any~$M$ and is a
special $B$-field in oriented superstring theory.  It may be identified with
the generator of $ ko^{-2} (\pt;\RZ)\cong KO^{-2} (\pt;\RZ)\cong \zt $.  Of
course, $\ce$~maps to~$\eta $ under the Bockstein homomorphism
$R^{-2}(\pt;\RZ)\to R^{-1}(\pt;\ZZ)$.

Any real line bundle~$L\to M$ determines
  \begin{equation}\label{eq:52}
     \ce [L]\in R^{-2}(M;\RZ)\longrightarrow \cR^{-1}(M)
  \end{equation}
with $t\bigl(\ce[L] \bigr)=1$ and $a\bigl(\ce [L] \bigr)=w_1(L)$; see
Lemma~\ref{thm:29}.  

        \begin{remark}[]\label{thm:38}
 An oriented superstring spacetime~$X^{10}$ is endowed with a spin
structure~$\kappa $.  (The twisted notion of spin structure for superstring
orientifold spacetimes is the subject of~\S\ref{sec:6}.)  Now the
$B$-field~$\cb$ may be written (Lemma~\ref{thm:29}) as a sum of an
object~$\cb_0$ in~$\cH^3(X)$ and a $\zt$-graded double cover $K\to X$, the
latter with characteristic class $\bigl(t(\cb),a(\cb) \bigr)\in
H^0(X;\zt)\times H^1(X;\zt)$.  We can shuffle the data: Define \emph{two}
spin structures $\kappa _\ell =\kappa $, $\kappa _r=\kappa +K$ on spacetime
and consider the $B$-field to be~$\cb_0$.  The two spin structures then
correlate with the two spin structures~$\alpha _\ell ,\alpha _r$ on the
worldsheet; see Definition~\ref{thm:2} below.  This splitting into `left' and
`right' does \emph{not} generalize to orientifolds.
        \end{remark}

 \section{The $B$-field amplitude and worldsheet spin structures}\label{sec:2}
% lastsubsec@  2

The spacetime for oriented bosonic string theory is a smooth 26-manifold~$X$,
and the $B$-field~$\cb $ has an equivalence class in~$\cH^3(X)$; see
Supposition~\ref{thm:6}.  The worldsheet in oriented bosonic string theory is
a closed 2-manifold~$\Sigma $ with orientation~$\mathfrak{o}$ and a smooth
map $\phi \:\Sigma \to X$.  (It represents the propagation of closed strings;
for open strings $\Sigma $~may have a boundary.)  One factor in the
exponentiated action of the worldsheet theory is
  \begin{equation}\label{eq:12}
     \exp\left( 2\pi i\int_{\Sigma }\phi ^*\cb \right);
  \end{equation}
it only depends on the equivalence class ~$\beq\in \cH^3(X)$ and is defined
using the pushforward in ordinary differential cohomology: $\phi ^*\beq\in
\cH^3(\Sigma )$ and the orientation~$\mathfrak{o}$ on~$\Sigma $ determines a
pushforward map~\cite[\S3.5]{HS}
  \begin{equation}\label{eq:13}
     \int_{(\Sigma ,\mathfrak{o})}\:\cH^3(\Sigma )\longrightarrow
     \cH^1(\pt)\cong \RZ.  
  \end{equation}
In this section we define the analog for the superstring and explore some
consequences.

 \subsection{Spin structures on superstring worldsheets}\label{subsec:2.1}

As a preliminary we quickly review spin structures.  Recall that the
intrinsic geometry of a smooth $n$-manifold~$M$ is encoded in its principal
$GL_n\RR$-bundle of frames $\sB(M)\to M$.  A point of~$\sB(M)$ is a linear
isomorphism $\RR^n\to T_mM$ for some~$m\in M$.  Choose a Riemannian metric
on~$M$, equivalently, a reduction to an $O_n$-bundle of frames $\sB_O(M)\to
M$.  The spin group
  \begin{equation}\label{eq:14}
     \rho \:\Spin_n\longrightarrow O_n 
  \end{equation}
is the double cover of the index two subgroup $SO_n\subset O_n$.  A
\emph{spin structure} on~$M$ is a principal $\Spin_n$-bundle $\sstr\to M$
together with an isomorphism of the associated $O_n$-bundle with~$\sB_O(M)$.
It induces an orientation on~$M$ via the cover $\Spin_n\to SO_n$.  The space
of Riemannian metrics is contractible, so a spin structure is a topological
choice and can alternatively be described in terms of a double cover of an
index two subgroup of~$GL_n\RR$.  An isomorphism of spin structures is a map
$\sstr\to \sstr'$ such that the induced maps on $O_n$-bundles commutes with
the isomorphisms to~$\sB_O(M)$.  The \emph{opposite spin structure} to
$\sstr\to M$ is the complement of~$\sstr$ in the principal $\Pinm_n$-bundle
associated to the inclusion $\Spin_n\hookrightarrow \Pinm_n$;
see~\cite[Lemma~1.9]{KT} for more elaboration.\footnote{Recall that
$\Pin_n^{\pm}$~sits in the Clifford algebra~$\Cliff_n^{\pm}$ whose generators
satisfy~$\gamma ^2=\pm1$.  Either sign can be used to construct the opposite
spin structure.}  If $M$~admits spin structures, then the collection of spin
structures forms a groupoid whose set of equivalence classes~$\Spinstr(M)$ is
a torsor for $H^0(M;\zt)\times H^1(M;\zt)$; the action of a function $\delta
\:\pi _0M\to\zt$ in~$H^0(M;\zt)$ sends a spin structure to its opposite on
components where~$\delta =1$ is the nonzero element.  The automorphism group
of any spin structure is isomorphic to~$H^0(M;\zt)$; a function $\delta \:\pi
_0M\to\zt$ acts by the central element of~$\Spin_n$ on components
where~$\delta =1$.  The manifold~$M$ admits spin structures if and only if
the Stiefel-Whitney classes~$w_1(M),w_2(M)$ vanish.

A superstring worldsheet~ $(\Sigma, \mathfrak{o}) $ is oriented and is
equipped with a pair of spin structures\footnote{`$\ell $' and~`r' stand for
`left' and `right'.}~$\alpha \mstrut_\ell,\alpha _r$ which induce opposite
orientations at each point.  Our convention is that the left spin
structure~$\alpha _\ell $ induces the chosen orientation~$\mathfrak{o}$.
Observe that a spin structure is \emph{local} and can be considered as a
field in the sense of physics.  It is a discrete field, in fact a finite
field on a compact manifold: there are only finitely many spin structures up
to isomorphism.  As with gauge fields, spin structures have automorphisms so
there is a groupoid of fields rather than a space of fields.

        \begin{definition}[]\label{thm:2}
 The \emph{topological data on an oriented superstring
worldsheet}~$(\Sigma,\mathfrak{o}) $ is a discrete field~$\alpha $ which on
each connected orientable open set~$U\subset \Sigma $ is a pair of spin
structures which induce opposite orientations of~$U$.
        \end{definition}

\noindent
 In more detail, this is the indicated data on each connected orientable open
set, isomorphisms of the spin structures on intersections of such open sets,
and a coherence condition among the isomorphisms on triple intersections.
The global orientation~$\mathfrak{o}$ is used to construct from~$\alpha $ a
global spin structure~$\alpha \mstrut_\ell$ which induces~$\mathfrak{o}$ and
a spin structure~$\alpha _r$ which induces the opposite
orientation~$-\mathfrak{o}$.  The global spin structures~$\alpha
\mstrut_\ell,\alpha _r$ need not be opposites (as defined in the previous
paragraph).  For orientifold models~(\S\ref{sec:3}) the worldsheet does not
have a global orientation, indeed may be nonorientable, but it retains the
discrete field~$\alpha $; see Definition~\ref{thm:16}.  In string theory one
integrates over~$\alpha $, i.e., sums over the spin structures.

        \begin{remark}[]\label{thm:47}
 We could, of course, replace~$\alpha $ in Definition~\ref{thm:2} with the
pair of spin structures~$\alpha _\ell ,\alpha _r$.  Our formulation
emphasizes both the \emph{local} nature of the spin structure and that this
local field is the same on worldsheets in orientifold superstring theories.
        \end{remark}

 \subsection{Superstring $B$-field amplitudes}\label{subsec:2.2}

Let $X$~be a 10-manifold---a superstring spacetime---and $\cb $~a $B$-field
on~$X$ as defined in~\eqref{eq:9}.  We define the oriented superstring
$B$-field amplitude~\eqref{eq:12}, which only depends on the equivalence
class~$\beq\in \cR\inv (X)$.  To do so we replace~\eqref{eq:13} with a
pushforward in differential $R$-theory.  The main point is that the
cohomology theory~$R$ is Spin-oriented, that is, there is a pushforward in
topological $R$-theory on spin manifolds.  It is the Postnikov truncation of
the pushforward in $ko$-theory defined from the spin structure (which by the
Atiyah-Singer index theorem has an interpretation as an index of a Dirac
operator).  In fact, because we are in sufficiently low dimensions we can
identify it exactly with the pushforward in~$ko$, a fact which is useful in
the proof of the Theorem~\ref{thm:3} below.  Combining with integration of
differential forms we obtain a pushforward~\cite[\S4.10]{HS}
  \begin{equation}\label{eq:15}
     \int_{\Sigma ,\alpha \mstrut_\ell}\:\cR^{-1}(\Sigma )\longrightarrow
     \cR^{-3}(\pt)\cong \RZ 
  \end{equation}
in differential $R$-theory defined using the spin structure~$\alpha \mstrut_\ell$
on~$\Sigma $.  (Use ~\eqref{eq:10} to see the isomorphism $\cR^{-3}(\pt)\cong
\RZ$.)  This completes the definition of the $B$-field amplitude.  In the
remainder of this section we investigate special cases which go beyond the
$B$-field amplitude for the oriented bosonic string.
 
Let $(\Sigma, \mathfrak{o}) $~be a closed oriented surface and
$\Spinstr(\Sigma ,\mathfrak{o})$ the set of equivalence classes of spin
structures which refine the given orientation.  Note $\Spinstr(\Sigma
,\mathfrak{o})$~ is a torsor for~$H^1(\Sigma ;\zt)$.  Let
  \begin{equation}\label{eq:16}
     q\:\Spinstr(\Sigma,\mathfrak{o} )\longrightarrow \zt 
  \end{equation}
be the affine quadratic function which distinguishes even and odd spin
structures.  It dates back to Riemann and is the Kervaire invariant in
dimension two; see~\cite[\S1]{HS} for some history.  The characteristic
property of the quadratic function~$q$ is
  \begin{equation}\label{eq:17}
     q(\alpha + a_1 + a_2) - q(\alpha +a_1) -q(\alpha +a_2) + q(\alpha ) =
     a_1\cdot a_2,\qquad \alpha \in \Spinstr(\Sigma,\mathfrak{o} ),\quad a_1,a_2\in
     H^1(\Sigma ;\zt), 
  \end{equation}
where $a_1\cdot a_2\in \zt$ is the mod~2 intersection pairing.

        \begin{theorem}[]\label{thm:3}
 Let $\ce $~be the nonzero universal $B$-field in~\eqref{eq:11}.  For
any superstring worldsheet $\phi \:\Sigma \to X$, the $B$-field amplitude
is~$(-1)^{q(\alpha \mstrut_\ell)}$.
        \end{theorem}

\noindent 
 This demonstrates  that the $B$-field amplitude~\eqref{eq:12} is sensitive to
the worldsheet spin structure.

        \begin{proof}
 Let $p\:\Sigma \to\pt$ and $p_*^{\al}\:ko^0(\Sigma ;\ZZ)\to
ko^{-2}(\pt;\ZZ)$ the pushforward~\eqref{eq:15} defined using the spin
structure~$\alpha \mstrut_\ell$.  .  Since~\cite[\S4.10]{HS} pushforward is
compatible with the exact sequence~\eqref{eq:10}, we use push-pull to compute
the integral in~\eqref{eq:12} as
  \begin{equation}\label{eq:18}
      p^{\al}_*p^*\ce =\ce
     p^{\al}_*(1).
  \end{equation}
The main theorem in~\cite{A1} states that $p^{\al}_*(1)=q(\alpha
\mstrut_\ell)\eta ^2$, where $\eta ^2\in ko^{-2}(\pt;\ZZ)\cong \zt$ is the
generator.  Finally, $\ce\cdot \eta ^2\in ko^{-4}(\pt;\RZ)\cong \RZ$ is the
nonzero element~$1/2$ of order two~ \cite[Proposition~B.4]{FMS}.
        \end{proof}

The space of fields~$\sF$ in the worldsheet formulation has many components,
distinguished by the equivalence class of the spin structures~$\alpha $, the
homotopy class of~$\phi \:\Sigma \to X$, etc.  If $\cb$~is any $B$-field
on~$X$, then Theorem~\ref{thm:3} implies that the theory with $B$-field
$\cb+\ce $ differs only by the sign~$(-1)^{q(\alpha \mstrut_\ell)}$ on
components of~$\sF$ with spin structure~$\alpha \mstrut_\ell$.  Note that
$t(\cb+\ce) = t(\cb)+1$.  Recall the notation in~\eqref{eq:19}.

        \begin{definition}[]\label{thm:4}
 An oriented superstring has \emph{Type~IIB} on components of~$X$ on which
~$t(\cb)\:\pi _0X\to\zt$ vanishes and has \emph{Type~IIA} on components
of~$X$ on which $t(\cb)$~is nonzero.
        \end{definition}

        \begin{remark}[]\label{thm:39}
 In the Hamiltonian formulation the distinction between Type~IIA and Type~IIB
is a sign in the GSO projection.  In the Lagrangian formulation this sign is
manifested by the sign~$(-1)^{q(\alpha _\ell )}$ in the sum over spin
structures~\cite{SW}.  Also, since the set of isomorphism classes of
$B$-fields is an abelian group there is a distinguished element, namely
zero. In this sense our approach favors Type~IIB as more ``fundamental'' than
Type~IIA.
        \end{remark}

Next, we consider the worldsheet amplitude for the special flat $B$-fields
defined in~\eqref{eq:52}.

        \begin{theorem}[]\label{thm:30}
 Let $L\to X$ be a real line bundle and $\ce L$ the corresponding
$B$-field.  For a superstring worldsheet $\phi \:\Sigma \to X$, the $B$-field
amplitude is~$(-1)^{q(\alpha \mstrut_\ell + \phi ^*L)}$.
        \end{theorem}

        \begin{proof}
 We proceed as in the proof of Theorem~\ref{thm:3}.  The right hand side
of~\eqref{eq:18} is now~$\ce p_*[\phi ^*L]$.  Conclude by observing that the
pushforward of~$[\phi ^*L]$ in the spin structure~$\alpha \mstrut_\ell$ is
equal to the pushforward of~1 in the spin structure $\alpha \mstrut_\ell+\phi
^*L$.
        \end{proof}

Lemma~\ref{thm:29} implies that $t\bigl(\ce [L] \bigr)=1$ and
$a\bigl(\ce[L] \bigr)= w_1(L)$.  We can consider instead the
$B$-field~$\ce (L-1)$ for which~$t=0$ and $a$~is as before; then
combine Theorem~\ref{thm:3} and Theorem~\ref{thm:30} to compute the $B$-field
amplitude 
  \begin{equation}\label{eq:61}
     (-1)^{q(\alpha _\ell +\phi ^*L) - q(\alpha _\ell )} 
  \end{equation}
for the $B$-field $\ce(L-1)$.

  \section{Orbifolds and orientifolds}\label{sec:3}
% lastsubsec@  3

In this section we take up two important variations of the basic Type~II
superstring.  First, suppose a finite group~$\Gamma $ acts on a smooth
10-manifold~$Y$.  Then there is a superstring theory---the
\emph{orbifold}---whose spacetime is constructed from the pair~$(Y,\Gamma )$
by ``gauging'' the symmetry group~$\Gamma $.  The main new feature is the
inclusion of \emph{twisted sectors}~\cite{DHVW}: in addition to strings $\phi
\:\cir\to Y$ one considers for each~$\gamma \in \Gamma $ maps $\phi \:\RR\to
Y$ such that $\phi (s+1)=\gamma \cdot \phi (s)$ for all~$s \in \RR$.  The
analog for surfaces is a bit more complicated.  Twisted sectors are labeled
by a principal $\Gamma $-bundle $P\to\Sigma $ over a superstring
worldsheet~$\Sigma $, and then a map to spacetime is a $\Gamma $-equivariant
map $\tilde\phi \:P\to Y$.  If $\tilde\phi '\:P'\to Y$ is another orbifold
worldsheet, then a morphism $\tf\to \tf'$ is an isomorphism $P\to P'$ of
principal $\Gamma $-bundles which intertwines~$\tf,\tf'$.  The space of these
fields is an infinite-dimensional groupoid.
 
Points of~$Y$ connected by elements of~$\Gamma $ represent the same points of
spacetime---$\Gamma $ is a gauge symmetry---so it is natural to take
spacetime as the quotient~$Y\gpd\Gamma $.  We keep track of isotropy
subgroups, due to non-identity elements~$\gamma \in \Gamma $ and~$y\in Y$
with~$\gamma \cdot y=y$.  Now an old construction in differential
geometry~\cite{Sa}, also dubbed~\cite{Th} `orbifold', does exactly that.
Furthermore, we can admit as spacetimes orbifolds~$X$ which are not global
quotients by finite groups, thus widening the collection of models introduced
in the previous paragraph.  Orbifolds are presented by a particular class of
\emph{groupoids}\footnote{We could write `orbifold'=`smooth Deligne-Mumford
stack', smooth understood as in `smooth manifold'.}~\cite{ALR}, a special
case being the presentation of a global quotient~$X=Y\gpd\Gamma $ by the
pair~$(Y,\Gamma )$.  We take up groupoid presentations in subsequent papers,
but here simply work directly with~$X$.  A worldsheet is then a map $\phi
\:\Sigma \to X$ of orbifolds, and the infinite-dimensional orbifold of such
maps includes twisted sectors.  The reader unfamiliar with
differential-geometric orbifolds may prefer to consider only global
quotients~$Y\gpd \Gamma $ and work equivariantly on~$Y$.

 \subsection{Equivariant cohomology and orbifold $B$-fields}\label{subsec:3.1}

There are many extensions of a given cohomology theory~$h$ to an equivariant
cohomology theory for spaces~$Y$ with the action of a compact Lie group~$G$.
The simplest is the \emph{Borel construction}.  It attaches to~$(Y,G)$ the
space~$Y_G = EG\times _GY$, where $EG$~is a contractible space with a free
$G$-action.  Then one defines the Borel equivariant $h$-cohomology as
$h_G(Y):=h(Y_G)$.  This is not a new cohomology theory, but rather the
nonequivariant theory applied to the Borel construction, a functor from
$G$-spaces to spaces.  That functor generalizes to orbifolds which are not
necessarily global quotients---the functor is \emph{geometric
realization}---and so leads to a notion of ``Borel cohomology'' theories on
orbifolds.  But usually $h$~has other extensions to an equivariant theory.
For example, the Atiyah-Segal geometric version of equivariant $K$-theory,
defined in terms of equivariant vector bundles, is more delicate: Borel
equivariant $K$-theory appears as a certain completion~\cite{AS}.  The
Atiyah-Segal theory is extended to orbifolds, in fact to ``local quotient
groupoids'', in~\cite{FHT}.
 
We recalled at the beginning of~\S\ref{sec:1} that the charges and fluxes
associated to an abelian gauge field in a quantum gauge theory lie in
generalized cohomology groups.  When we pass to theories formulated on
orbifolds we must additionally specify a flavor of equivariant cohomology to
locate the charges and fluxes.  For example, the Ramond-Ramond field in
superstring theory has charges and fluxes in $K$-theory.  In the
corresponding orbifold theory they are in Atiyah-Segal equivariant
$K$-theory.  This choice has consequences even locally, at the level of
differential forms: it is consistent with extra Ramond-Ramond fields in
twisted sectors.  We hope to elaborate in a future paper.  Here we limit
consideration to $B$-fields on orbifolds.

Let $M$~be a 26-dimensional orbifold.  We posit the following generalization
of Supposition~\ref{thm:6}. 

        \begin{supposition}[]\label{thm:5}
 For the oriented bosonic orbifold the flux of the $B$-field~$\cb$ lies in
the Borel cohomology~$H^3(X;\ZZ)$.
        \end{supposition}

\noindent
 Furthermore, there is a generalization of differential cohomology to
orbifolds~ \cite{LU, G}.  So an immediate reformulation locates the
$B$-field itself in orbifold differential cohomology (see~\eqref{eq:9}).
Supposition~\ref{thm:5} is implicit in the literature, for example
in~\cite{Sh, GSW}.  The $B$-field amplitude~\eqref{eq:12} is defined
as before; the integration is still over a smooth manifold, the
worldsheet~$\Sigma $.
 
For the superstring case we also posit Borel cohomology for the $B$-field.
Let $X$~be a 10-dimensional orbifold. 

        \begin{proposal}[]\label{thm:7}
 For the superstring orbifold the flux of the $B$-field~$\cb$ lies in
the Borel cohomology~$R^{-1}(X)$.
        \end{proposal}

\noindent
 We are not aware of any general equivariant version of generalized
differential cohomology, much less a version for orbifolds.  In~\cite{DFM2}
we develop a geometric model of~$\cR\inv (X)$ for a local quotient
groupoid~$X$ and locate the $B$-field there.  The pullback to a worldsheet
then lives in the differential $R$-theory as in the non-orbifold case, and the
amplitude~\eqref{eq:12} is defined as before.

 \subsection{Orientifolds and $B$-fields}\label{subsec:3.2}

The \emph{orientifold} construction applies to both the bosonic string and
the superstring.  In its simplest incarnation the construction involves a
pair~$(Y,\sigma )$ of a smooth manifold~$Y$ and an involution $\sigma \:Y\to
Y$.  Fields on~$Y$ have a definite transformation law under~$\sigma $.  For
example, the metric is invariant whereas the 3-form field strength~$H$ of the
$B$-field is anti-invariant: $\sigma ^*H=-H$.  We combine the orbifold and
this simple orientifold by starting with a triple~$(Y,\Gamma ,\upsilon )$
consisting of a finite group~$\Gamma $, a smooth $\Gamma $-manifold~$Y$, and
a surjective homomorphism $\upsilon \:\Gamma \to\zt$.  Then fields on~$Y$
transform under~$\Gamma $: e.g., the 3-form field strength of the $B$-field
satisfies
  \begin{equation}\label{eq:20}
     \gamma ^*H = (-1)^{\upsilon (\gamma )}H,\qquad \gamma \in \Gamma .  
  \end{equation}
As before $\Gamma $~acts as a gauge symmetry and the physical points of
spacetime lie in the quotient.  Therefore, we arrive at a more general model
in a geometric formulation.  

        \begin{definition}[]\label{thm:8}
 The spacetime of an orientifold string model is an orbifold~$X$ equipped
with a double cover of orbifolds $\pi \:\Xw\to X$.
        \end{definition}

\noindent
 The equivalence class $w\in H^1(X;\zt)$ of the double cover lies in the
Borel cohomology of~$X$.  For the triple~$(Y,\Gamma ,\upsilon )$ the double
cover is $\pi \:Y\gpd \ker\upsilon \to Y\gpd\Gamma $ with characteristic class
in~$H^1_\Gamma (Y;\zt)$.

Definition~\ref{thm:8} applies to both the bosonic string and the
superstring.  There is a particular special case of the orientifold
construction which goes back to the early superstring theory literature.

        \begin{definition}[]\label{thm:11}
 The \emph{Type~I superstring} on a smooth 10-manifold~$Y$ is the orientifold
with spacetime $X=Y\times \pt\gpd(\zt)$ the orbifold quotient of the trivial
involution on~$Y$.  
        \end{definition}

We next generalize Supposition~\ref{thm:5} and Proposal~\ref{thm:7} to
bosonic and superstring orientifolds.  First, recall that if $M$~is any space
and $A\to M$ a fiber bundle of discrete abelian groups, then we can define
twisted ordinary cohomology~$H^{\bullet }(M;A)$ with coefficients in~$A$.  In
particular, if $M_w\to M$ is a double cover, then we form the associated
bundle $A_w\to M$ of free abelian groups of rank one, defined by the action
of~$\{\pm1\}$ on~$\ZZ$.  We denote the associated twisted cohomology by
$H^{w+\bullet }(M;\ZZ)$.  It has a concrete manifestation in terms of cochain
complexes: the deck transformation of the double cover $M_w\to M$ acts on the
cochain complex~$C^{\bullet }(M_w;\ZZ)$, and $H^{w+\bullet }(M;\ZZ)$ is the
cohomology of the anti-invariant subcomplex.  If $M$~is a smooth manifold
there is a corresponding twisted version~$\cH^{w+\bullet }(M)$ of
differential cohomology.  We use the model of differential cohomology as a
cochain complex of triples~$(c,h,\omega )$, where $c\in C^{\bullet
}(M_w;\ZZ)$, $\omega \in \Omega ^{\bullet }(M_w)$, and $h\in C^{\bullet
+1}(M_w;\RR)$ (see~\cite[\S6.3]{DF}, \cite[\S2.3]{HS}), and take the
anti-invariant subcomplex.

        \begin{supposition}[]\label{thm:12}
 Let $X_w\to X$ be a double cover of 26-dimensional orbifolds and suppose
$X$~is the spacetime of a bosonic orientifold.  Then the flux of the
$B$-field~$\cb$ lies in the twisted Borel cohomology $H^{w+3}(X;\ZZ)$.
        \end{supposition}

\noindent
 This appears in the literature using a different model of twisted degree
three cohomology~\cite{GSW}.  The equivalence class of the $B$-field lies in
the twisted differential cohomology group~$\cH^{w+3}(X)$, consistent with the
transformation law~\eqref{eq:20}.

The $B$-field quantization law for the \emph{superstring} orientifold is
expressed in terms of twisted $R$-cohomology.  The following discussion
applies to \emph{any} cohomology theory~$h$.  Let $M_w\to M$ be a double
cover of a space~$M$ with deck transformation~$\sigma $, and as
after~\eqref{eq:2} let $\{h_p\}_{p\in \ZZ}$ denote a spectrum representing
$h$-cohomology. Recall that $h^p(M)$~is the abelian group of homotopy classes
of maps $M\to h_p$.  Let $i_p\:h_p\to h_p$ be a map which represents the
additive inverse on cohomology classes, and we may assume $i_p\circ
i_p=\id_{h_p}$.  Define a $w$-twisted $h$-cocycle of degree~$p$ on~$M$ to be a
pair~$(c,\eta )$ of a map $c\:M_w\to h_p$ and a homotopy $\eta $ from $\sigma
^*c$ to $i_pc$.  A homotopy of $w$-twisted $h$-cocycles is a $w$-twisted
$h$-cocycle on~$\Delta ^1\times M$, where $\Delta ^1$~is the 1-simplex.  Then
$h^{w+p}(M)$~is defined as the group of homotopy classes of $w$-twisted
$h$-cocycles of degree~$p$.  A small elaboration using triples as
in~\eqref{eq:8} defines $w$-twisted $\chh$-cohomology if $M$~is a smooth
manifold.  In~\cite{DFM2} we develop a differential-geometric model
for~$\cR^{w-1}(M)$.

        \begin{proposal}[]\label{thm:13}
 Let $X_w\to X$ be a double cover of 10-dimensional orbifolds and suppose
$X$~is the spacetime of a superstring orientifold.  Then the flux of the
B-field~$\cb$ lies in the twisted Borel cohomology~$R^{w-1}(X)$. 
        \end{proposal}

        \begin{remark}[]\label{thm:40}
 There is an important restriction on the $B$-field flux which we will derive
in~\S\ref{sec:6}.  Namely, a superstring orientifold spacetime~ $X$ carries a
suitably twisted spin structure defined in terms of the $B$-field, and its
existence leads to the constraints~\eqref{eq:45}, \eqref{eq:47}.
        \end{remark}

 \subsection{Universal $B$-fields on orientifolds}\label{subsec:3.3}

Let $B\zt=\pt\gpd(\zt)$ and $\pi _0\:\pt\to B\zt$ the universal double cover,
which we denote~$w_0$.  The geometric realization of~$B\zt$
is~$\RP^{\infty}$, so the Borel $R$-cohomology of~$B\zt$ is the
$R$-cohomology of~$\RP^{\infty}$.  For orientifolds there are universal
$B$-fields pulled back from the classifying map $X\to B\zt$ of the
orientifold double cover $X_w\to X$.  For the bosonic orientifold we first
apply the exact sequence analogous to~\eqref{eq:10},
  \begin{equation}\label{eq:22}
     0\longrightarrow H^{w+2}(M;\RZ)\longrightarrow
     \cH^{w+3}(M)\longrightarrow \Omega _\ZZ^{w+3}(M)\longrightarrow 0 ,
  \end{equation}
to $M=B\zt$ and deduce $\cH^{w_0+3}(B\zt)\cong H^{w_0+2}(B\zt;\RZ)$.  Now the
twisted chain complex of the geometric realization~$\RP^{\infty}$ is
  \begin{equation}\label{eq:24}
 \xymatrix{
      \ZZ&\ZZ\ar[l]_{2}&\ZZ\ar[l]_{0}
      &\ZZ\ar[l]_{2}&\ZZ\ar[l]_{0}&\ar[l]_{}\cdots}
  \end{equation}
Apply $\Hom(-,\RZ)$ to compute 
  \begin{equation}\label{eq:25}
     \cH^{w_0+3}(B\zt;\ZZ)\cong H^{w_0+2}(B\zt;\RZ)\cong \zt. 
  \end{equation}
This is the universal group of $B$-fields on bosonic orientifolds.

        \begin{remark}[]\label{thm:15}
 The Bockstein map $H^{w_0+2}(B\zt;\RZ)\to H^{w_0+3}(B\zt;\ZZ)$ is an isomorphism,
as follows easily from the long exact sequence associated to
$\ZZ\to\RR\to\RZ$.  This is also obvious from the geometric picture of
differential cohomology given around~\eqref{eq:26} since in this case
$\cH^{w_0+3}(B\zt)$~ is finite, hence equal to its group of components~$H^{w_0+3}
(B\zt;\ZZ)$. 
        \end{remark}

For superstring orientifolds we also have a finite group of universal twistings.

        \begin{theorem}[]\label{thm:14}
 The group $\cR^{w_0-1}(B\zt)\cong R^{w_0-2}(B\zt;\RZ)\cong R^{w_0-1}(B\zt;\ZZ)$ is
cyclic of order~8.  For any generator~$\cth$ we can identify~$4\cth$ with the
nonzero element in~\eqref{eq:25}.  Furthermore, the pullback of~$\cth$ under
$\pi _0\:\pt\to B\zt$ is~$\ce$.
        \end{theorem}

\noindent 
 Recall that $\ce$~is the nonzero class in~\eqref{eq:11}.  In~\cite{DFM3} we
interpret~$R^{w_0-1}(B\zt;\ZZ)$ as the group of universal twistings of
$KO$-theory modulo Bott periodicity, which may be identified with the super
Brouwer group~\cite[p.~195]{Wa}, \cite[Proposition~3.6]{De}.  

        \begin{proof}
 All cohomology groups in this proof have $\ZZ$~coefficients.  We first show
  \begin{equation}\label{eq:53}
     R^{w_0-1}(B\zt):= R^{w-1}(\RP^{\infty})\cong R^{w-1}(\RP^4)\cong 
     ko^{w-1}(\RP^4),
  \end{equation}
where `$w$'~denotes the nontrivial double cover of projective space.  The
first equality is the definition of (twisted) Borel cohomology.  The second
group is computed as the space of sections of a twisted bundle of spectra
over~$\RP^{\infty}$ whose fiber is~$R_{-1}$; see~\cite{ABGHR, MS}.
The second isomorphism follows from elementary obstruction theory since
$R_{-1}$~has vanishing homotopy groups above degree~3; see~\eqref{eq:3}.
Finally, the $(-1)$-space of the $ko$-spectrum and~$R_{-1}$ have the same
5-skeleton, which justifies the final isomorphism in~\eqref{eq:53}.
 
Write $ko^{w-1}(\RP^4)\cong ko_{\zt}^{w_0-1}(S^4)$.  Here we use the
Atiyah-Segal equivariant $ko$-theory for the antipodal action on the sphere;
the equivariant double cover~$w_0$ is pulled back from a point.  Next, we
claim
  \begin{equation}\label{eq:54}
     ko^{w_0-1}_{\zt}(\pt)\cong ko^0(\pt) .
  \end{equation}
For in the Atiyah-Bott-Shapiro (ABS) model with Clifford algebras~\cite{ABS},
the left hand side is the $K$-group of a category of $\zt$-graded real
modules for the $\zt$-graded algebra~$A$ generated by odd elements~$\gamma
,\alpha $ with $\gamma ^2=-1$, $\alpha ^2=1$, and $\alpha \gamma =-\gamma
\alpha $.  (That the generator~$\alpha $ of~$\zt$ is odd reflects the
twisting~$w_0$; the Clifford generator~$\gamma $ is always odd.)  But $A$~is
isomorphic to the $\zt$-graded matrix algebra~$\End(\RR^{1|1})$, and so the
category of $A$-modules is Morita equivalent to the category of $\zt$-graded
real vector spaces.  Let\footnote{We reserve the notation~`$\xi $' for the
inverse class in twisted \emph{periodic} $KO$-theory.  It is the $KO$-Euler
class of the real line with involution~$-1$, viewed as an equivariant line
bundle over a point.  It has many beautiful properties, some of which we
exploit in~\cite{DFM3}.} $\xi \inv $~denote the element
in~$ko^{w_0-1}_{\zt}(\pt)$ which corresponds to~$1\in ko^0(\pt)$ under the
isomorphism~\eqref{eq:54}.  In the ABS~model $\xi \inv $~ is represented by
  \begin{equation}\label{eq:58}
     \xi \inv :\quad \RR^{1|1}\textnormal{ with }
\gamma =\left(\begin{matrix} 0&-1\\1&0 \end{matrix}\right),\quad 
     \alpha =\left(\begin{matrix} 0&1\\1&0 \end{matrix}\right) .
  \end{equation}
Then multiplication by~$\xi \inv $ induces an isomorphism
$\tko^{0}_{\zt}(S^4)\cong ko^{w_0-1}_{\zt}(S^4)$, where the tilde denotes
reduced $ko$-theory.  Now~$\tko^{0}_{\zt}(S^4)\cong \tko^0(\RP^4)$ and
$\tko^0(\RP^4)$~is cyclic of order~8 generated by~$H-1$, where $H\to\RP^4$ is
the nontrivial (Hopf) real line bundle: the order of~$\tko^0(\RP^4)$ is
bounded by~8 by the Atiyah-Hirzebruch spectral sequence, and because
$w_4\bigl(4(H-1) \bigr)\not= 0$ we conclude $4(H-1)\not= 0$.
 
The assertion about~$4\cth$ follows from the twisted version of the exact
sequence~\eqref{eq:4} on~$B\zt$: the kernel group~$H^{w_0+3}(B\zt;\ZZ)$
is~\eqref{eq:25}.  To prove the last statement we observe that the argument
in the previous paragraph identifies the generator of~$ko^{w_0-1}_{\zt}(S^4)$
as the pullback of~$\xi \inv $ under the $\zt$-equivariant map
$h\:S^4\to\pt$.  Let $i\:\pt\hookrightarrow S^4$ be the (nonequivariant)
inclusion of a point.  Then $\pi _0^*(\cth)$ is the image of~$h^*\xi \inv $
under the composition $ko^{w_0-1}_{\zt}(S^4)\to
ko^{-1}(S^4)\xrightarrow{i^*}ko^{-1}(\pt)$, which is evidently the image
of~$\xi \inv $ under $ko_{\zt}^{w_0-1}(\pt)\to ko^{-1}(\pt)$.  (We choose
orientations of~$\pt$ and~$S^4$ to trivialize the pullback of~$w_0$
under~$\pi _0$.)  Finally, in the ABS model this pullback simply drops the
action of~$\alpha $, and what remains of~\eqref{eq:58} is the generator~$\eta
$ of~$ko^{-1}(\pt)\cong \zt$.
        \end{proof}

  \section{The $B$-field amplitude for orientifolds}\label{sec:4}
% lastsubsec@  4

A worldsheet in an orientifold string theory has several
fields~\cite[Definition~5]{DFM1}.  For the bosonic case they all appear in
Definition~\ref{thm:9}; for the superstring there are additional fields
articulated in Definition~\ref{thm:16} and Definition~\ref{thm:22}.

 \subsection{Bosonic orientifold worldsheets}\label{subsec:4.4}

As a preliminary recall that a smooth $n$-manifold~$M$ has a canonical
orientation double cover $\hp\:\hM\to M$ defined as the quotient~$\hM
:=\sB(M)/GL^+_n\RR$, where $GL^+_n\RR$~is the group of orientation-preserving
automorphisms of~$\RR^n$.  The manifold~$\hM$ is canonically oriented.  It is
natural to denote the double cover~$\hp\:\hM\to M$ as~`$w_1(M)$'.

        \begin{definition}[]\label{thm:9}
 Let $\pi \:\Xw\to X$ be the spacetime of an orientifold string theory.
An \emph{orientifold worldsheet} is a triple~$(\Sigma ,\phi ,\tf)$ consisting
of a compact 2-manifold~$\Sigma $, a smooth map $\phi \:\Sigma \to X$, and an
equivariant lift $\tf\:\hS\to \Xw$ of~$\phi $ to the orientation double cover
of~$\Sigma $. 
        \end{definition}

\noindent
 In theories with open strings $\Sigma $~may have nonempty boundary.  The
surface~$\Sigma $ is not oriented and need not be orientable.  In fact, the
existence of the equivariant lift implies a constraint involving its first
Stiefel-Whitney class:
  \begin{equation}\label{eq:21}
     \phi ^*w=w_1(\Sigma );
  \end{equation}
the equivariant lift~$\tf$ is an isomorphism of the double covers
in~\eqref{eq:21}.\footnote{In our ambiguous notation `$w$'~and
`$w_1(M)$'~denote both a double cover and its equivalence class.}  For an
orientifold spacetime defined by a triple~$(Y,\Gamma ,\upsilon )$ as above,
Definition~\ref{thm:9} unpacks to a principal $\Gamma $-bundle $P\to \Sigma
$, an orientation on~$P$, and a $\Gamma $-equivariant map $P\to Y$.  There is
a constraint: if $\upsilon (\gamma )=0$, then the action of~$\gamma $ on~$P$
preserves the orientation; if $\upsilon (\gamma )=1$, then $\gamma $~reverses
the orientation.  There is an obvious notion of equivalence of
triples~$(\Sigma ,\phi ,\tf)$, and the collection of such triples forms a
groupoid presentation of an infinite dimensional orbifold.

        \begin{remark}[]\label{thm:10}
  Definition~\ref{thm:9} applied to a single string clarifies the nature of
twisted sectors in orientifold theories.  Namely, if $\phi \:\cir\to X$ is a
string, then the constraint implies that $\phi ^*w=0$, since the circle is
orientable.  Thus $\phi $~lifts to the double cover~$\Xw$.  Put differently,
the homotopy class of~$\phi $ does not detect a nontrivial double cover, so
does not sense the orientifold.  Now the ``twisting'' in a twisted sector for
a global quotient orbifold~$X=Y\gpd \Gamma $ measures the extent to which a
string $\cir\to X$ fails to lift to a string $\cir\to Y$.  So if
$X=Y\gpd\Gamma $ is a global quotient with $\upsilon \:\Gamma \to\zt$
specifying the orientifold, then $\phi $~lifts to~$\Xw=Y\gpd\ker\upsilon $ and
the twisted sectors are labeled by conjugacy classes in~$\ker\upsilon $.  In
case $\Xw=Y$~is a smooth manifold and $X$~the orbifold quotient by an
involution, then any string $\phi \:\cir\to X$ lifts to a loop $\cir\to Y$.
Hence there are no twisted sectors in a ``pure'' orientifold.
        \end{remark}

 \subsection{$B$-field amplitudes for bosonic orientifolds}\label{subsec:4.1}

Recall that if $M$~is a smooth compact $n$-manifold then integration of
differential forms
  \begin{equation}\label{eq:27}
     \int_{M,\mathfrak{o}}\:\Omega ^n(M)\longrightarrow \RR 
  \end{equation}
is only defined after choosing an orientation~$\mathfrak{o}$.  Absent an
orientation one may integrate densities, which in our current notation are
$w_1$-twisted differential forms: forms on the orientation double
cover~$\hM$ which are odd under the deck transformation.  Integration of
densities is a homomorphism
  \begin{equation}\label{eq:28}
     \int_{M}\:\Omega ^{w_1(M)+n}(M)\longrightarrow \RR 
  \end{equation}
which lifts to integration in twisted differential cohomology: 
  \begin{equation}\label{eq:29}
     \int_{M}\:\cH ^{w_1(M)+n+1}(M)\longrightarrow \cH^{1}(\pt)\cong \RZ . 
  \end{equation}
To define~\eqref{eq:29} one may follow~\cite[\S3.4]{HS} working in the model
with smooth singular cochains.
 
That understood, the definition of the $B$-field amplitude~\eqref{eq:12} for
bosonic orientifolds is straightforward.  Let $\cb$~be a bosonic orientifold
$B$-field as in Supposition~\ref{thm:12}; its equivalence class is $\beq\in
\cH^{w+3}(X)$.  Then for an orientifold worldsheet as in
Definition~\ref{thm:9} the isomorphism~\eqref{eq:21} (defined by~$\tf$ in
Definition~\ref{thm:9}) places the pullback~$\phi ^*\beq$ in the
group~$\cH^{w_1(\Sigma )+3}(\Sigma )$.  The $B$-field amplitude is then
computed using a twisted integration~\eqref{eq:29} in place of~\eqref{eq:13}.
This bosonic orientifold $B$-field amplitude is described using a particular
model for~$\cH^{w+3}(X)$ in~\cite{GSW}.
 
The universal $B$-field amplitude is easy to compute.

        \begin{proposition}[]\label{thm:18}
 Let $\cb$~be the nonzero universal $B$-field in~\eqref{eq:25}.  Then for any
bosonic orientifold worldsheet the $B$-field amplitude~\eqref{eq:12} is
$(-1)^{\Euler(\Sigma )}$, where $\Euler(\Sigma )$~is the Euler number of the
closed surface~$\Sigma $.
        \end{proposition}

       \begin{proof}
 If $\phi \:\Sigma \to X$ is the worldsheet map, then we can identify $\phi
^*[\cb]\in H^{w_1(\Sigma )+2}(\Sigma ;\RZ)$ as the pullback of~$x^2\in
H^2(\RP^{\infty};\zt)$ via the map $w_1\:\Sigma \to\RP^{\infty}$ which
classifies~$w_1(\Sigma )$.  The latter pulls back the generator $x\in
H^1(\RP^{\infty};\zt)$ to~$w_1(\Sigma )$, so $\phi ^*\beq=w_1(\Sigma )^2$.
Now $w_1(\Sigma )^2=w_2(\Sigma )$ since the difference of the two sides is
the second Wu class, which vanishes on manifolds of dimension less than four.
Finally, $w_2(\Sigma )$~is the mod~2 reduction of the Euler class (which in
general lives in twisted integral cohomology).
        \end{proof}

 \subsection{Spin structures on superstring orientifold
worldsheets}\label{subsec:4.2} 

Turning to the worldsheet in a superstring orientifold theory we begin by
specifying the appropriate notion of spin structure.  We could not find this
definition in the string theory literature, even for the Type~I superstring.

        \begin{definition}[]\label{thm:16}
  The \emph{topological data on a superstring orientifold worldsheet}~$\Sigma
$ is a discrete field~$\alpha $ which on each connected orientable open
set~$U\subset \Sigma $ is a pair of spin structures which induce opposite
orientations of~$U$.
        \end{definition}

\noindent
 Definition~\ref{thm:16} is identical to Definition~\ref{thm:2} except for
the omission of the global orientation.  Although $\alpha $~is locally a pair
of spin structures, there is no global spin structure on~$\Sigma $.  Rather,
the local pair of spin structures with opposite underlying orientation glue
to a global spin structure on the orientation double cover~$\hS$.  The global
description is equivalent to the local Definition~\ref{thm:16}, and we
use~`$\alpha $' to denote the spin structure on~$\hS$ as well as the local
field in Definition~\ref{thm:16}.  Let $\hs$~denote the involution on~$\hS$.
If the spin structures are locally opposite consistent with gluing---more
simply, if the pullback~$\hs^*\alpha$ of the global spin structure on~$\hS$
is the opposite~$-\alpha $---then a refinement to a \pinm\ structure
on~$\Sigma $ may be possible, but is additional data.

        \begin{remark}[]\label{thm:28}
 The oriented double cover~$S^2$ of~$\RP^2$ has a unique spin structure (up
to~$\cong $) compatible with the orientation.  It refines to two inequivalent
\pinm\ structures on~$\RP^2$.  On the other hand, the oriented double
cover~$\cir\times \cir$ of the Klein bottle~$K$ has 4~inequivalent spin
structures compatible with the orientation.  Two of them each refine in two
inequivalent ways to give four inequivalent \pinm\ structures on the Klein
bottle; the other two each refine in two inequivalent ways to give four
inequivalent \pinp\ structures on the Klein bottle.
        \end{remark}

        \begin{remark}[]\label{thm:24}
 It is important to emphasize that for general~$\alpha $ there is no
refinement to a \pinm\ structure.  (Indeed, if $\alpha$ refines to a \pinm\
structure then the pullback to the orientation double cover defines an
\emph{equivariant} spin structure.)  This has important ramifications for the
physics. Consider a connected open set $U\subset \Sigma$ with the topology of
a cylinder. On $U$ there are four choices of a pair of spin structures: each
spin structure can be either bounding or non-bounding when restricted to the
circle. In the case where one spin structure bounds, and the other does not,
it is impossible to refine~$\alpha$ to a \pinm\ structure since the pullback
of the pair to the oriented double cover of~$U$ is not invariant under the
deck transformation.  From the physical viewpoint, it is clear from the
Hamiltonian formulation of the string theory that this mixed choice of spin
structures occurs for Feynman diagrams in which spacetime fermions propagate
along an internal line corresponding to $U$.  Conversely, restricting
attention to only those~ $\alpha$ which do refine to a \pinm\ structure
misses all of the sectors of the worldsheet theory in which space-time
fermions propagate along that channel.
        \end{remark}

        \begin{remark}[]\label{thm:17}
 Consider an orientifold theory in which the orientifold double cover $\pi
\:\Xw\to X$ is trivial \emph{and trivialized}.  Then Definition~\ref{thm:16}
reduces to Definition~\ref{thm:2}.  For the trivialization may be modeled as
a section of~$\pi $.  Then for an orientifold worldsheet
(Definition~\ref{thm:9}) $\phi \:\Sigma \to X$ the equivariant lift~$\tf$
identifies $\phi ^*(\Xw\to X)\cong (\hS\to \Sigma )$, and so the section
of~$\pi $ pulls back to a section of $\hp\:\hS\to\Sigma $.  But the latter is
precisely a global orientation~$\mathfrak{o}$ of~$\Sigma $.
        \end{remark}

 \subsection{$B$-field amplitudes for superstring orbifolds}\label{subsec:4.3}

To describe the $B$-field amplitude~\eqref{eq:12} for the superstring we need
the analog of~\eqref{eq:29} in differential $R$-theory.  A complete
definition involves twistings of cohomology theories beyond twists by double
covers (see the discussion preceding Proposal~\ref{thm:13}) and is deferred
to~\cite{DFM2}.  For now recall that $R$~is Spin-oriented and there is a
pushforward \eqref{eq:15} on spin manifolds.  More generally, the obstruction
to a spin structure on an $n$-manifold~$M$ determines a twisting~$\tR M$ of
$R$-theory, so too of differential $R$-theory, and a twisted pushforward
  \begin{equation}\label{eq:32}
     \int_{M}\:\cR^{\,\tR M-3}(M)\longrightarrow \cR^{-3}(\pt)\cong \RZ. 
  \end{equation}
The twisting~$\tR M$ includes the dimension of~$M$, as well as the
Stiefel-Whitney classes~$w_1(M),w_2(M)$.  A spin structure produces an
isomorphism $n\to \tR M$ and so reduces the pushforward~\eqref{eq:32} to a
pushforward on untwisted differential $R$-theory, as in~\eqref{eq:15}.
 
Now suppose $\pi\:\Xw\to X$ is the orientifold double cover of a
10-dimensional superstring spacetime~$X$ with $B$-field~$\cb$.  Given a
worldsheet as in Definitions~\ref{thm:9} and~\ref{thm:16} the pullback $\phi
^*\beq$ of the equivalence class of the $B$-field lies in $\cR^{w_1(\Sigma
)-1}(\Sigma )$.  It seems, then, that to push forward to a point
using~\eqref{eq:32} we need an isomorphism $w_1(\Sigma )\to \tR \Sigma -2$ of
twistings of $R$-theory.  However, the local spin structures~$\alpha $
on~$\Sigma $---equivalently global spin structure on~$\hS$---do not give such
an isomorphism.  This puzzle stymied the authors for a long period.  The
resolution is that the $B$-field amplitude in general is not a number, but
rather an element in a complex line:
  \begin{equation}\label{eq:59}
     \textnormal{\emph{The $B$-field amplitude for a superstring orientifold is
     anomalous.}}\footnote{We refer to a term in an (effective) action as
     \emph{anomalous} if it takes values in a (noncanoncially trivialized)
     complex line rather than the complex numbers.}
  \end{equation}
There is one case in which there is an isomorphism $w_1(\Sigma )\to \tR
\Sigma -2$, namely when $\alpha $~is refined to a \pinm\ structure on~$\Sigma
$.  Then the $B$-field amplitude may be defined as a number.  Notice that on
a \pinm\ worldsheet the two local spin structures~$\alpha $ are opposites.
The anomaly measures the extent to which that fails for general~$\alpha $.

        \begin{remark}[]\label{thm:20}
 To illustrate, suppose that the superstring orientifold worldsheet $\Sigma
$~is diffeomorphic to a 2-dimensional torus.  Even though $\Sigma $~is
orientable, the fields do not include an orientation.  The field~$\alpha $ is
equivalent to a pair of spin structures~$\alpha ',\alpha ''$ on~$\Sigma $
with opposite underlying orientations.  Up to isomorphism there are 4~choices
for each of~$\alpha ',\alpha ''$, so 16~possibilities in total.  Of those
4~refine uniquely to \pinm\ structures on~$\Sigma $.  The $B$-field
amplitudes for the remaining~12 are anomalous.
        \end{remark}

Recall from Theorem~\ref{thm:3} that in the oriented case the universal
$B$-field amplitude for the superstring computes the well-known $\zt$-valued
quadratic form on spin structures.  We now investigate the analogous
amplitude in the orientifold case for \pinm worldsheets.  Let $\Sigma $~be a
closed 2-manifold and $\sPm(\Sigma )$ the $H^1(\Sigma ;\zt)$-torsor of
equivalence classes of \pinm\ structures.  Let $\cth $~be a generator of the
cyclic group~$R^{w_0-2}(B\zt;\RZ)$; see Theorem~\ref{thm:14}.  Now the
orientation double cover determines a map $h\:\Sigma \to B\zt$ and so a class
$h^*\cth \in R^{w_1(\Sigma )-2}(\Sigma ;\RZ)$.  Let $p\:\Sigma \to\pt$.  Then
a \pinm\ structure~$\alpha ^-$ on~$\Sigma $ determines a pushforward map
  \begin{equation}\label{eq:55}
     p_*^{\am}\:R^{w_1(\Sigma )-2}(\Sigma ;\RZ)\longrightarrow
     R^{-4}(\pt;\RZ)\cong \RZ. 
  \end{equation}
Define
  \begin{equation}\label{eq:33} 
  \begin{aligned}
     q^-\:\sPm(\Sigma )&\longrightarrow \RZ \\ 
     \am &\longmapsto p_*^{\am}(h^*\cth) 
  \end{aligned}
  \end{equation}
We can replace the $R$-cohomology groups in~\eqref{eq:55} with $ko$-groups or
even periodic $KO$-groups.

        \begin{theorem}[]\label{thm:37}
 The function~$q^-$ takes values in~$\frac 18\ZZ/\ZZ\cong \zmod8$, is a
quadratic refinement of the intersection pairing, and its reduction modulo
two is congruent to the Euler number~$\Euler(\Sigma )$. 
        \end{theorem}

        \begin{proof}
 The first statement follows since~$8\cth=0$.  We must show that for
$a_1,a_2\in H^1(\Sigma ;\zt)$,
  \begin{equation}\label{eq:56}
     q^-(\am + a_1 + a_2) - q^-(\am +a_1) -q^-(\am +a_2) + q^-(\am ) = 
     \frac 12a_1\cdot a_2,\qquad \am \in \sPm(\Sigma ).
  \end{equation}
The argument of~\cite[p.~53]{A1} applies verbatim through Lemma~(2.3), which
we replace with the following assertion. Let $i\:\pt\hookrightarrow \Sigma $
and $u=i_*(\eta ^2)\in ko^{0}(\Sigma ;\ZZ)$; then 
  \begin{equation}\label{eq:57}
     p^{\am}_*(h^*\cth\cdot u)=1/2.
  \end{equation}
To prove this we note that $u$~is supported in a neighborhood of a point
in~$\Sigma $, so by excision we can compute the left side on a sphere~$S^2$.
Fix an orientation of~$S^2$, which is a section of the orientation double
cover~$w_1(\Sigma )$.  This lifts $h\:S^2\to B\zt$ to $\pi _0\:\pt\to B\zt$.  Then
since by Theorem~\ref{thm:14} we have~$\pi _0^*\cth=\ce$, we
reduce~\eqref{eq:57} to $p_*(\ce\cdot u)$, which by push-pull is~$\ce\cdot
\eta ^2$.  As in the proof of Theorem~\ref{thm:3} this is nonzero.

The last statement follows from Proposition~\ref{thm:18} since $4\cth$~is the
nonzero element of~\eqref{eq:25}; see Theorem~\ref{thm:14}.
        \end{proof}

Recall~\cite[\S3]{KT} that the \pinm\ bordism group~$\Opinm$ is cyclic of
order eight and the Kervaire invariant is an isomorphism. 

        \begin{corollary}[]\label{thm:43}
 With an appropriate choice of generator~$\cth$ in Theorem~\ref{thm:14}, the
quadratic form~\eqref{eq:33} is the Kervaire invariant.
        \end{corollary}

\noindent
 For oriented surfaces the $\zt$-valued Kervaire invariant~\eqref{eq:16} has
a well-known $KO$-theoretic interpretation~\cite{A1}.  Corollary~\ref{thm:43}
provides a similar $KO$-theoretic interpretation in the unoriented case. 

        \begin{proof}
 The definition~\eqref{eq:33} of~$q^-$ is evidently a bordism invariant. The
real projective plane~$\RP^2$ has two \pinm\ structures; either
generates~$\Opinm$.  Since $\RP^2$ has odd Euler number, the value of~$q^-$
on either \pinm\ structure is a generator of~$\zmod8$.  The four possible
choices of~$\cth$ in the definition of~$q^-$ give the four generators
of~$\zmod8$, so we can choose the one which matches the standard Kervaire
invariant on~$\RP^2$, hence on all \pinm\ surfaces.
        \end{proof}

  \section{Worldsheet fermions and spacetime spin structures}\label{sec:5}
% lastsubsec@  3

A fermionic functional integral is, by definition, the pfaffian of a Dirac
operator.  It is naturally an element of a line, so in a family of bosonic
fields a section of a line bundle over the parameter space~\cite[Part~2]{F1}.
For an orientifold superstring worldsheet the $B$-field amplitude is also
anomalous~\eqref{eq:59}.  The main result of~\cite{DFM2} is that the product
of these anomalies is trivializable, and furthermore the correct notion of
spin structure on spacetime~(\S\ref{sec:6}) leads to a trivialization.  In
this section, after identifying the fermionic fields in the 2-dimensional
worldsheet theory, we work out an analogous phenomenon in a familiar
1-dimensional theory: the ``spinning particle''.  Namely, in
Theorem~\ref{thm:27} we identify the pfaffian line of the Dirac operator on a
circle in terms of the frame bundle of spacetime, and show how a spin
structure on spacetime leads to a trivialization.

 \subsection{Fermions on orientifold superstring worldsheets}\label{subsec:5.2}

This is the last in the triad of definitions (see Definitions~\ref{thm:9}
and~\ref{thm:16}) specifying the fields on an orientifold superstring
worldsheet~\cite[Definition~5]{DFM1}.

        \begin{definition}[]\label{thm:22}
 An \emph{orientifold superstring worldsheet} consists of $(\Sigma ,\phi
,\tf,\alpha )$ as in Definitions~\ref{thm:9} and~\ref{thm:16} together with a
positive chirality spinor field~$\psi $ on~$\hS$ with coefficients
in~$\hp^*\phi ^*TX$ and a negative chirality spinor field~$\chi $ on~$\hS$
with coefficients in~$T^*\hS$. 
        \end{definition}

\noindent
 The notion of chirality is defined by the canonical orientation on the
orientation double cover~$\hS$; the spinors use the spin structure~$\alpha $.
Both~$\psi $ (the ``matter fermion'') and~$\chi $ (the ``gravitino'') should
be regarded as local fields on~$\Sigma $, but the global description on~$\hS$
is more transparent; the action is local on~$\Sigma $.  The crucial factor in
the functional integral over~$\psi ,\chi $ for fixed~$\phi $ and~$\alpha $ is
the pfaffian of a Dirac operator on~$\hS$, which may be written
  \begin{equation}\label{eq:34}
     \pfaff D_{\hS,\alpha }\bigl(\hp^*\phi ^*TX-T\Sigma  \bigr). 
  \end{equation}
The pfaffian line bundle is local, so we can \emph{heuristically} analyze it
on a small contractible open set~$U\subset \Sigma $.  Now $\pi \inv U\subset
\hS$~is the disjoint union of two oppositely oriented open sets diffeomorphic
to~$U$ with spin structures~$\alpha ',\alpha ''$ refining the underlying
orientations.  The pfaffian~\eqref{eq:34} is anomalous on each component
of~$\pi \inv U$.  If the spin structures~$\alpha ',\alpha ''$ are opposite,
then the product of the anomalies is trivializable; an isomorphism of~$\alpha
'$ with the opposite of~$\alpha ''$ trivializes the anomaly.  So we see that
the anomaly measures the failure of~$\alpha '$ and~$\alpha ''$ to be
opposites, just as for the $B$-field.\footnote{The anomaly also depends on
the topology of~$\phi ^*(TX)$.}  (See the text leading to
Remark~\ref{thm:20}.)
 
For the oriented superstring a global argument for the triviality of the
pfaffian line bundle---the anomaly in the fermionic functional
integral~\eqref{eq:34}---is given in~\cite[\S4]{FW}.  In the non-orientifold
case there is no anomaly in the $B$-field amplitude (see~\eqref{eq:15}).  The
argument in~\cite{FW} only proves the triviality; it does not provide a
trivialization so does not determine a definition of~\eqref{eq:34} as a
function.  (This additional data is sometimes termed a `setting of the
quantum integrand'.)  In fact, the superstring data \emph{does} determine a
trivialization: it is the \emph{spacetime} spin structure which is critical.
We explore this two-dimensional anomaly problem in~\cite{DFM2} and show that
the trivialization varies under a change of spacetime spin structure. 

        \begin{remark}[]\label{thm:45}
 For an oriented superstring worldsheet (Definition~\ref{thm:2}), the
dependence is as follows.  Suppose $a\in H^1(X;\zt)$ is a change of spacetime
spin structure and $b= \alpha_l - \alpha_r\in H^1(\Sigma ;\zt)$ the
difference of the two global worldsheet spin structures.  Then the
trivialization for a worldsheet $\phi \:\Sigma \to X$ multiplies
by 
  \begin{equation}\label{eq:63}
     (-1)^{\langle \phi ^*a,b \rangle} 
  \end{equation}
where $\langle -,- \rangle$~is the $\zt$-valued pairing on~$H^1(\Sigma
;\zt)$.  Combining this factor with~\eqref{eq:61} one sees that our
formulation of the oriented superstring has the expected left-right symmetry.
See~\eqref{eq:62} for a 1-dimensional analog of~\eqref{eq:63}.
Equation~\eqref{eq:63} is consistent with~\cite{AW}.
        \end{remark}

 \subsection{A supersymmetric quantum mechanical theory}\label{subsec:5.3}

Here we illustrate the impact of the spacetime spin structure on the
worldsheet pfaffian in a simpler quantum field theory: the 1-dimensional
supersymmetric quantum mechanical system whose partition function computes
the index of the Dirac operator~\cite{Ag, FW, W1}.  In this
theory spacetime~$X$ is a Riemannian manifold of arbitrary dimension~$n$.
For the classical theory it does not have a spin structure or even an
orientation.  However, to simplify we assume that $X$~is oriented.  The
worldsheet of superstring theory is replaced by a 1-dimensional manifold~$S$
with a map $\phi \:S\to X$.  The manifold~$S$ is endowed with a single spin
structure.  The fermionic fields of Definition~\ref{thm:22} are replaced by a
single spinor field~$\psi $ on~$S$ with coefficients in~$\phi ^*TX$.
 
Consider $S=\cir$ with the nonbounding spin structure~$\alpha $.  The first
step in computing the partition function is to compute the fermionic
functional integral over~$\psi $ for a fixed loop $\phi \:\cir\to X$, which
is the pfaffian
  \begin{equation}\label{eq:36}
     \pfaff D_{\cir,\alpha }(\phi ^*TX). 
  \end{equation}
As the Dirac operator on the circle is real, the square of its pfaffian line
bundle is canonically trivial and so the square of~\eqref{eq:36} is a
well-defined function.  There is a standard regularization and the result
(see~\cite{A3}, for example) is
  \begin{equation}\label{eq:37}
     \bigl(\pfaff D_{\cir,\alpha }(\phi ^*TX)\bigr)^2 = \det\bigl(1-\hol(\phi
     )\bigr) , 
  \end{equation}
where $\hol(\phi )\in SO_n$ is the holonomy, well-defined up to conjugacy.
We may as well assume that $n$~is even, or else \eqref{eq:37}~vanishes
identically.  Now the function $g\mapsto \det(1-g)$ on~$SO_n$ does not have a
smooth square root.  However, its lift to~$\Spin_n$ does have a square
root~$f$, the difference of the characters of the half-spin representations:
  \begin{equation}\label{eq:38}
     f(\tg) =  i^{n/2}\bigl(\chi \mstrut _{\Delta ^+}(\tg) - \chi \mstrut
     _{\Delta ^-}(\tg)\bigr),\qquad \tg\in \Spin_n.  
  \end{equation}
Hence given a spin structure on~$X$ we can lift the holonomy function
$\hol\:LX\to SO_n$ on the loop space of~$X$ to a function
$\thol\:LX\to\Spin_n$, and so define~\eqref{eq:36} as
  \begin{equation}\label{eq:39}
     \pfaff D_{\cir,\alpha }(\phi ^*TX) := f\bigl(\thol(\phi ) \bigr). 
  \end{equation}
The right hand side of~\eqref{eq:39} manifestly uses the spin structure on
spacetime~$X$.  Note that we can equally replace the function~$f$ by its
negative; the overall sign is not determined by this argument.

        \begin{remark}[]\label{thm:41}
 If we change the spin structure on~$X$ by a class $a\in H^1(X;\zt)$, then it
follows immediately from~\eqref{eq:39} that the pfaffian multiplies by 
  \begin{equation}\label{eq:62}
     (-1)^{\phi ^*(a)[\cir]}. 
  \end{equation}
        \end{remark}

% \subsection{Pfaffian line bundles and trivializations}\label{subsec:5.1}

The pfaffian is more naturally an element of a line and for the analogy with
the 2-dimensional worldsheet theory it is more illuminating to analyze the
pfaffian line ~$\Pfaff D_{\cir,\alpha }(\phi ^*TX)$ directly.
(See~\cite[\S3]{F2} for the definition of the pfaffian line of a Dirac
operator.)  Write $E\to\cir$ for the oriented vector bundle~$\phi ^*TX$.  The
Dirac operator~$D_{\cir,\alpha }$ is the covariant derivative~$\nabla$ acting
on sections of~$E\to \cir$.  It is real and skew-adjoint, so its pfaffian
line is real.  Identify a real line~$L$ with the $\zt$-torsor $\pi
_0\bigl(L\setminus \{0\} \bigr)$, so obtain the pfaffian torsor ~$\Pfaff
\nabla$.  Let $\sBSO(E)\to\cir$ denote the bundle of oriented orthonormal
frames of~$E$.  It is trivializable since $SO_n$~is connected.  The space of
sections~$\Gamma $ has two components and is naturally a torsor for~$\pi
_1(SO_n)\cong \zt$.  Furthermore, a spin structure $\sstr\to\sBSO(E)\to\cir$
trivializes the torsor~$\pi _0\Gamma $: there is a distinguished component of
sections which lift to~$\sstr$.

        \begin{theorem}[]\label{thm:27}
 There is a canonical isomorphism $\Pfaff \nabla \cong \pi _0\Gamma $.
Therefore, a spin structure on~$E$ determines a trivialization of~$\Pfaff
\nabla$.
        \end{theorem}

\noindent
 Suppose $Z$~is any manifold and $E\to Z\times \cir$ an oriented bundle with
covariant derivative.  Then the Pfaffian torsors vary smoothly in~$z\in Z$ so
form a double cover of~$Z$.  Its characteristic class may be computed from
the Atiyah-Singer index theorem as the slant product~$w_2(E)/[\cir]$;
see~\cite[(5.22)]{FW}.  Theorem~\ref{thm:27} is a ``categorification'' of
this topological result---an isomorphism of line bundles rather than simply
an equality of their isomorphism classes--- necessary in order to discuss
trivializations.  We remark that more sophisticated categorifications of the
Atiyah-Singer index theorem are needed for anomaly problems in higher
dimensions, such as~\cite{DFM2}; see~\cite{Bu} for a recent result in
dimension two.

        \begin{proof}
 Fix a Riemannian metric on~$\cir$ of total length~1.  The covariant
derivative of a framing $e\in \Gamma $ is a function $\nabla
(e)\:\cir\to\son$.  Using parallel transport choose~$e$ so that $\nabla
(e)$~is a constant skew-symmetric matrix~$A$ whose eigenvalues~$a\sqrt{-1}$
satisfy $-\pi <a\le\pi $.  Note that $\exp(A)$~is the holonomy of~ $\nabla$.
The framing~$e$ is determined up to a constant element of~$SO_n$.  In
particular, the span~$W$ of the basis vectors of~$e$ in the space~$\sH$ of
sections of~$E\to \cir$ is independent of this choice.  It is easy to see
that $\nabla $~is invertible on the orthogonal complement~$W^\perp$ to~$W$
in~$\sH$ relative to the $L^2$~metric.  So $\Pfaff\nabla $~ is canonically
the determinant line~$\Det W^*$ of the finite dimensional vector space~$W^*$,
and the associated $\zt$-torsor is canonically the $\zt$-torsor~$\sT$ of
orientation classes of bases of~$W$.  But a basis of~$W$ is an element
of~$\Gamma $, so $\sT$~is canonically~$\pi _0\Gamma $, as claimed.
        \end{proof}

        \begin{remark}[]\label{thm:42}
 Formula~\eqref{eq:62} for the change of trivialization as a function of the
change of spin structure on~$E$ follows immediately: $E\to\cir$ has two spin
structures and they determine two different points of~$\pi _0\Gamma $. 
        \end{remark}

  \section{The twisted spin structure on a superstring orientifold
spacetime}\label{sec:6}  
% lastsubsec@000

The spacetime~$X$ of an oriented superstring theory has a spin structure.
There is a modification for orientifolds in superstring theory: the notion of
spin structure is twisted by both the orientifold double cover $\pi \:\Xw\to
X$ and the $B$-field.  In this section we describe this twisted notion of
spin structure in concrete differential-geometric terms.
 
Recall quite generally that if $\rho \:G\to G'$ is a homomorphism of Lie
groups and $P\to M$ a principal $G$-bundle over a space~$M$, then there is an
associated principal $G'$-bundle $\rho (P)\to M$, defined by the ``mixing
construction'' $\rho (P)= P\times _GG'$.  Conversely, if $Q\to M$ is a
principal $G'$-bundle, then a \emph{reduction to~$G$ along~$\rho $} is a
pair~$(P,\varphi )$ consisting of a principal $G$-bundle $P\to M$ and an
isomorphism $\varphi \:\rho (P)\to Q$.  If $M^n$~is a smooth manifold and
$\rho \:G\to GL_n\RR$, then a reduction of the $GL_n\RR$ frame
bundle~$\sB(M)$ to~$G$ along~$\rho $ is called a \emph{$G$-structure} on~$M$.
We defined orientations in these terms in~\ref{subsec:4.4} and spin
structures in these terms in~\ref{subsec:2.1}; for convenience we used a
metric and so a homomorphism~\eqref{eq:14} into the orthogonal group.  A
principal $G$-bundle is classified by a map\footnote{More precisely, a
classifying map for $P\to M$ is a $G$-equivariant map $P\to EG$ for $EG\to
BG$ a universal $G$-bundle.}  $M\to BG$ whose homotopy class is an invariant
of $P\to M$.  The topological classification of reductions along $\rho \:G\to
G'$ may be analyzed as a lifting problem:
  \begin{equation}\label{eq:40}
     \xymatrix{&BG\ar[d]^{B\rho }\\M\ar@{-->}[ur]^P\ar[r]^Q & BG'}
  \end{equation}
Two particular cases are of interest here: (i)\ $\rho $~is the inclusion of
an index two subgroup, in which case $B\rho \:BG\to BG'$ is a double cover
and the obstruction to~\eqref{eq:40} lies in~$H^1(M;\zt)$; and (ii)\ $\rho
$~is a surjective double cover, in which case $B\rho \:BG\to BG'$ is a
principal $K(\zt,1)$-bundle\footnote{$K(\zt,1)$~is an Eilenberg-MacLane
space; a topological group model is the group of projective linear
transformations of an infinite dimensional real Hilbert space.} and the
obstruction to~\eqref{eq:40} lies in~$H^2(M;\zt)$.

The spin group~\eqref{eq:14} is a double cover of an index two subgroup
of~$O_n$.  We now define groups~$G_0,G_1$ which bear the same relation
to~$\tOn:=O_n\times \zt\times \zt$ via homomorphisms 
  \begin{equation}\label{eq:41}
     \rho _i\:G_i\longrightarrow \tOn,\qquad i=1,2
  \end{equation}
which factor through an index two subgroup $G_i'\subset \tOn$.  First, let
$D_4\to\zt\times \zt$ be the dihedral double cover in which the generators of
the $\zt$~factors lift to anticommuting elements of order two.  Define
$G_0,G_0'$ as the first two groups in
  \begin{equation}\label{eq:42}
     \rho _0\:(\Spin_n\times D_4)/\{\pm1\}\longrightarrow SO_n\times
     \zt\times \zt\longrightarrow \tOn, 
  \end{equation}
where $-1\in \{\pm1\}$ is the product of the central elements of~$\Spin_n$
and~$D_4$.  For~$G_1$ we first define the surjective homomorphism 
  \begin{equation}\label{eq:43}
     \begin{aligned} \tOn&\longrightarrow \zt \\ (g,a,b)&\longmapsto
      c+a,\qquad \det(g)=(-1)^c,\end{aligned} 
  \end{equation}
and let $G_1'$~be the kernel.  Then $G_1$ is the inverse image of~$G'_1$ under
  \begin{equation}\label{eq:44}
     (\Pinm_n\times D_4)/\{\pm1\}\longrightarrow \tOn.
  \end{equation}

Suppose $X$~is a superstring spacetime---a 10-dimensional orbifold---and
$\Xw\to X$ an orientifold double cover.  Proposal~\ref{thm:13} implies that a
$B$-field~$\cb$ is a geometric object whose equivalence class~$\beq$ lies
in~$\cR^{w-1}(X)$.  As in~\eqref{eq:19} there are topological invariants
$t(\cb)\:\pi _0X\to\zt$ and a double cover~$\Xa\to X$.

        \begin{definition}[]\label{thm:25}
 Let $\Xw\to X$ be the orientifold double cover of a Riemannian orbifold~$X$
which represents a superstring spacetime.  Let $\cb$~be a $B$-field on~$X$.
Then a \emph{twisted spin structure} is a reduction of the principal
$\tOt$-bundle
  \begin{equation}\label{eq:46}
     \sB\mstrut _O(X)\times \mstrut _X\Xw\times \mstrut _X\Xa\to X 
  \end{equation}
along $\rho \:G_i\to\tOt$, where $i\in \zt$ is chosen on each component
of~$X$ according to the value of~$t(\cb)$.
        \end{definition}

\noindent
 Definition~\ref{thm:4} expresses the two types in more familiar terms as
Type~IIB for~$t(\cb)=0$ and Type~IIA for~$t(\cb)=1$.  Typically spacetime is
connected and only one of these occurs.
 
Equivalence classes of twisted spin structures, if they exist, form a torsor
for $H^0(X;\zt)\times H^1(X;\zt)$.  The existence is settled by the
following. 

        \begin{proposition}[]\label{thm:26}
 Let $\Xw\to X$ and $\cb$~be as in Definition~\ref{thm:25}.  Then a twisted
spin structure exists if and only if 
  \begin{align}\label{eq:45}
      w_1(X) &= t(\cb)w \\ \label{eq:47} w_2(X) &= a(\cb)w +
      t(\cb)w^2
  \end{align} 
        \end{proposition}

\noindent
 These equations live in the Borel cohomology of the orbifold~$X$. 

        \begin{proof}
 Equation~\eqref{eq:45} is the condition to reduce the structure group
of~\eqref{eq:46} along the inclusion $G'_i\hookrightarrow \tOt$.  For
$G'_0=SO_{10}\times \zt\times \zt$ it is the condition $w_1(X)=0$ for an
orientation.  For~$t(\cb)=1$ the homomorphism~\eqref{eq:43} induces a
map~$B\tOt\to B\zt$ which pulls the generator of~$H^1(B\zt;\zt)$ back
to~$w_1+x$, where $H^1(B\tOt;\zt) = H^1(BO_{10};\zt)\times
H^1(B\zt;\zt)\times H^1(B\zt;\zt)$ has generators~$w_1,x,y$.
Then~\eqref{eq:45} follows by pullback along the classifying map
of~\eqref{eq:46}.

For~\eqref{eq:47} we first observe that the double cover $D_4\to\zt\times
\zt$ is classified by $xy\in H^2(B\zt\times B\zt;\zt)$.  Then the first
homomorphism in~\eqref{eq:42} induces a principal $K(\zt,1)$-bundle $BG_0\to
BG'_0$ classified by~$w_2+xy$, from which \eqref{eq:47}~follows on components
with~$t(\cb)=0$.  For components with~$t(\cb)=1$ we first
recall~\cite[Lemma~1.3]{KT} that the universal $K(\zt,1)$-bundle
$B\Pinm_{10}\to BO_{10}$ is classified by~$w_1^2+w_2\in H^2(BO_{10};\zt)$.
Then the definition~\eqref{eq:44} of~$G_1$ shows that $BG_1\to BG'_1$ is
classified by~$w_1^2+w_2 + xy$; equation ~\eqref{eq:47} now follows from this
and~\eqref{eq:45}.
        \end{proof}

        \begin{remark}[]\label{thm:46}
 The occurrence of~$D_4$ in our definition of a twisted spin structure is
closely related to the $D_4$ symmetry group\footnote{generated by the
worldsheet transformations $(-1)^{F_L}, (-1)^{F_R}$ and worldsheet
parity~$\Omega $} which appears in Hamiltonian treatments of orientifolds in
the physics literature.  We hope to elaborate on this elsewhere.
        \end{remark}

\bigskip\bigskip
\newcommand{\etalchar}[1]{$^{#1}$}
\providecommand{\bysame}{\leavevmode\hbox to3em{\hrulefill}\thinspace}
\providecommand{\MR}{\relax\ifhmode\unskip\space\fi MR }
\providecommand{\MRhref}[2]{%
  \href{http://www.ams.org/mathscinet-getitem?mr=#1}{#2}
}
\providecommand{\href}[2]{#2}

\end{document}